\documentclass[11pt]{article}

\usepackage{fullpage}
\usepackage{amsmath,amsthm,amssymb}
\usepackage{xspace}
\usepackage{graphicx}
\usepackage{xcolor}
\usepackage{subcaption}
\usepackage[colorlinks=true, citecolor=blue]{hyperref}
\usepackage{cleveref}
\usepackage{datetime}
\usepackage[ruled,vlined,linesnumbered]{algorithm2e}
\usepackage{authblk}
\usepackage{enumerate}
\usepackage{xspace}

\usepackage{mathtools}
\usepackage{csquotes}
\usepackage{tabularx}
\usepackage{skull}
\usepackage{subcaption}

\newtheorem{lemma}{Lemma}
\newtheorem{corollary}{Corollary}
\newtheorem{definition}{Definition}

\usepackage{cleveref}
\crefname{algocf}{alg.}{algs.}
\Crefname{algocf}{Algorithm}{Algorithms}

\DeclareMathOperator{\dd}{{\textup{d}}}
\DeclareMathOperator{\diam}{{\textup{diam}}}
\DeclareMathOperator{\radius}{{\textup{rad}}}
\DeclareMathOperator{\ecc}{{\textup{ecc}}}
\DeclareMathOperator{\OO}{{\mathcal{O}}}
\DeclareMathOperator{\mw}{mw}
\DeclareMathOperator{\sw}{sw}
\DeclareMathOperator{\pw}{pw}

\DeclareMathOperator{\bw}{bw}

\newcommand{\set}[1]{\protect\ensuremath{\left\{ #1\right\}}\xspace}

\def\FF{\ensuremath{\mathcal{F}}\xspace}

\begin{document}

\title{Enumeration of far-apart pairs by decreasing distance for faster hyperbolicity computation\thanks{This work has been supported by the French government, through the UCA$^\textsc{jedi}$ Investments in the Future project managed by the National Research Agency (ANR) with the reference number ANR-15-IDEX-01, and the Distancia project with reference number ANR-17-CE40-0015.}}

\author[1]{David Coudert}
\author[2]{Andr\'e Nusser}
\author[3]{Laurent Viennot}

\affil[1]{Universit\'e C\^ote d'Azur, Inria, CNRS, I3S, France}
\affil[2]{Max Planck Insitute for Informatics and Graduate School of Computer Science, Saarland Informatics Campus, Saarbrücken, Germany}
\affil[3]{Inria, Paris University, CNRS, Irif, France}

\date{\vspace{-5ex}}

\maketitle

\begin{abstract}
Hyperbolicity is a graph parameter which indicates how much the shortest-path distance metric of a graph deviates from a tree metric. It is used in various fields such as networking, security, and bioinformatics for the classification of complex networks, the design of routing schemes, and the analysis of graph algorithms.
Despite recent progress, computing the hyperbolicity of a graph remains challenging. Indeed, the best known algorithm has time complexity $O(n^{3.69})$, which is prohibitive for large graphs, and the most efficient algorithms in practice have space complexity $O(n^2)$. Thus, time as well as space are bottlenecks for computing the hyperbolicity.

In this paper, we design a tool for enumerating all far-apart pairs of a graph by decreasing distances. A node pair $(u, v)$ of a graph is far-apart if both $v$ is a leaf of all shortest-path trees rooted at $u$ and $u$ is a leaf of all shortest-path trees rooted at $v$. This notion was previously used to drastically reduce the computation time for hyperbolicity in practice. However, it required the computation of the distance matrix to sort all pairs of nodes by decreasing distance, which requires an infeasible amount of memory already for medium-sized graphs. We present a new data structure that avoids this memory bottleneck in practice and for the first time enables computing the hyperbolicity of several large graphs that were far out-of-reach using previous algorithms. For some instances, we reduce the memory consumption by at least two orders of magnitude. Furthermore, we show that for many graphs, only a very small fraction of far-apart pairs have to be considered for the hyperbolicity computation, explaining this drastic reduction of memory.

As iterating over far-apart pairs in decreasing order without storing them explicitly is a very general tool, we believe that our approach might also be relevant to other problems.\\

\noindent\textbf{Keywords:} Gromov hyperbolicity; graph algorithms, far-apart pairs iterator.
\end{abstract}

\newpage

\section{Introduction}
\label{sec:intro}

This paper aims at computing the hyperbolicity
of graphs whose size ranges from tens of thousands to millions of nodes. The hyperbolicity is a parameter of a metric space generalizing the idea of Riemannian manifolds with negative curvature. When considering the metric of a graph, it measures, to some extent, how much it deviates  
from a tree metric. This parameter was first introduced by Gromov in the context of automatic groups~\cite{Gromov1987} in relation with their Cayley graphs. 

Hyperbolicity has received great attention in computer science in the last decades as it seems to capture important properties of several large practical graphs such as Internet~\cite{ShavittT04}, the Web~\cite{MunznerB95} and databases relations~\cite{WalterR02}. It is also used to classify complex networks~\cite{Abu2015,AlrasheedD15,Kennedy2013} and was proposed as a measure of how much a network is ``democratic''~\cite{ADM14,BCC15}. Formal relationships between Gromov hyperbolicity and the existence of a core (a subset of vertices intersecting a constant fraction of all the shortest-paths) are investigated in~\cite{ChepoiDV17}. Reciprocally, the existence of a core is shown to be inherent to any hyperbolic network in~\cite{ChepoiDV17}.
Furthermore, small hyperbolicity has tractability implications and measuring hyperbolicity has applications in routing~\cite{Boguna2010,ChepoiDEHVX12,Krauthgamer2006}, approximating other graph parameters~\cite{Chepoi2008,DasGuptaKMY18} and bioinformatics~\cite{Chakerian2012,Dress2012}.
See~\cite{Abu2015,Dragan2013} for recent surveys.

Computing the hyperbolicity is often a prerequisite in the above applications. As hyperbolicity can be defined by a simple 4-point condition, it can be naively computed in $\OO(n^4)$ time. As far as we know, the best theoretical algorithm~\cite{FournierIV15} has time complexity $\OO(n^{3.69})$. Although its complexity is $o(n^4)$, it is still supercubic and the algorithm appears to be impractical for graphs with a few tens of thousands of nodes. On the lower bounds side it was shown that under the Strong Exponential Time Hypothesis~\cite{IPZ2001} hyperbolicity cannot be computed in subquadratic-time, even for sparse graphs~\cite{Borassi2014,CoudertD14,FournierIV15}.

The only practical algorithms that can manage larger graphs~\cite{BorassiCCM15,CohenCL15} enumerate all pairs of nodes by decreasing distance.
For each pair, each 4-tuple obtained with a previous pair is tested with regard to the 4-point condition defining hyperbolicity. Each test of a 4-tuple provides a lower bound of hyperbolicity.
This approach allows to stop the enumeration as soon as the distance of the scanned pair equals twice the best hyperbolicity lower bound found so far. As it scans a portion of all 4-tuples, its worst case complexity is $\OO(n^4)$ but it appears much faster in practice as first scanning pairs with large distances allows to find good lower bounds early in the enumeration. A main optimization for further reducing the number of pairs scanned, consists in considering only far-apart pairs, that is, pairs such that no neighbor of one node is further apart from the other node, see Section \ref{sec:def-notations} for a formal definition. It can be proven that the 4-point condition defining hyperbolicity holds on all 4-tuples if it holds on 4-tuples made up of two such far-apart pairs~\cite{Nogues2009,SotoGomez2011}. The main bottleneck of this method lies in its inherent quadratic space usage: all far-apart pairs and all-pair distances are stored in the current implementations of this approach.
Computing the hyperbolicity of practical graphs with millions of nodes thus remains a great challenge.

\subsection{Our approach}
To make progress towards this challenge, we propose to enumerate far-apart pairs by decreasing distance without computing all-pair distances. The key of our approach is to first compute all eccentricities, that is, for each node, what is the largest distance from it. Computing all eccentricities is feasible in practice, see Section \ref{sec:related_work}. Note that we obtain the diameter $D$ of the graph as a side product, as it simply is the maximum eccentricity. We then scan nodes with eccentricity $D$ and enumerate far-apart nodes at distance $D$ from them, that is, nodes at distance $D$ that form far-apart pairs with them. We then scan nodes with eccentricity at least $D-1$ and enumerate far-apart nodes at distance $D-1$ from them, and so on. We also include various optimizations proposed in~\cite{CohenCL15,BorassiCCM15} to further reduce the number of 4-tuples considered. The main difficulty of our approach is that we have to compute distances on the fly as we do not store all-pair distances. This mainly requires to perform a breadth-first search (BFS) for each node of a far-apart pair considered. Interestingly, we show how to prune these BFS searches  based on some of the optimizations proposed in~\cite{BorassiCCM15}.  Storing most recent BFS searches in a cache also increases performance for some instances as we observe some sharing of far-apart nodes in practice.

\subsection{Main contributions}

Our main contributions are the following.
\begin{itemize}
	\item We present the first non-naive algorithm for iterating over far-apart pairs that neither computes and stores all distances explicitly, nor sorts the node pairs by recomputing all distances from scratch whenever they are needed.
	\item This, for the first time, enables enumerating all far-apart pairs of large graphs. Previously this was not possible either due to excessive amounts of time or memory needed for the computation of all far-apart pairs.
	\item As the prime application of our algorithm for iterating over far-apart pairs, we significantly reduce the memory consumption when computing the graph hyperbolicity. The memory reduction is at least two orders of magnitude for some instances.
	\item This drastic memory reduction enables us to compute the hyperbolicity of many large graphs for the first time.
	\item Due to the significance of far-apart pairs in a graph (e.g., they are the defining vertices for radius, eccentricities, diameter, \dots), we believe that our contribution of a far-apart pair iterator is also relevant in other settings.
\end{itemize}

\subsection{Other related work} \label{sec:related_work}

The most advanced practical algorithm for computing hyperbolicity~\cite{BorassiCCM15} can be seen as a refinement of~\cite{CohenCL15} that further prunes the search space. A complementary approach proposed in~\cite{CohenCDL17} consists in splitting the graph further than biconnected components using clique decomposition.
Using such a decomposition could also be used in our framework to further reduce memory usage.

Practical algorithms for computing all eccentricities~\cite{TakesK13, DraganHV18, LiQQ+2018} came along in a line of research for improving diameter computation of real world graphs~\cite{AkibaIK15,CrescenziGHLM13,BorassiCHKMT15,DraganHV18}.

Several optimizations were proposed for BFS search. Most notably, \cite{Hagerup19} reduces space usage to $O(n)$ bits, and \cite{AkibaIY13} performs several BFS in parallel from a node and some of its neighbors using bit-parallel word operations. Both methods could be used to further improve our approach.

\subsection{Organization}
Definitions and notations used in this paper are introduced in \Cref{sec:def-notations}.
We then present our far-apart pair iterator in \Cref{sec:iterator}. \Cref{sec:hyperbolicity} is devoted to hyperbolicity. We recall its definition and review the best know algorithmic results on this parameter. We then present our memory efficient algorithm based on the proposed far-apart pair iterator.
In \Cref{sec:experimental}, we report on the experimental evaluation of the algorithms presented in this paper on various graphs. In particular, we compare our algorithm for computing hyperbolicity with the state-of-the-art algorithm~\cite{BorassiCCM15}. We conclude this paper in \Cref{sec:conclusion} with some directions for future research.

\section{Definitions and notations}
\label{sec:def-notations}

We use the graph terminology of~\cite{Bondy1976,Diestel1997}.
All graphs considered in this paper are finite, undirected, connected, unweighted and simple.
The graph $G=(V,E)$ has $n = |V|$ vertices and $m = |E|$ edges.
The open neighborhood $N_G(S)$ of a set $S \subseteq V$ consists of all vertices in $V \setminus S$ with at least one neighbor in $S$.

Given two vertices $u$ and $v$, a \emph{$uv$-path} of length $\ell \geq 0$ is a sequence of vertices $(u=v_0v_1\ldots v_{\ell}=v)$, such that $\{ v_i, v_{i+1} \}$ is an edge for every $i$.
In particular, a graph $G$ is \emph{connected} if there exists a $uv$-path for all pairs $u,v \in V$, and in such a case the \emph{distance} $\dd_G(u,v)$ is defined as the minimum length of a $uv$-path in $G$.
When $G$ is clear from the context, we write $\dd$ (resp. $N$) instead of $\dd_G$ (resp. $N_G$).
The \emph{eccentricity} $\ecc(u)$ of a vertex $u$ is the maximum distance between $u$ and any other vertex $v\in V$, i.e., $\ecc(u) = \max_{v\in V} \dd(u,v)$. The maximum eccentricity is the \emph{diameter} $\diam(G)$ and the minimum eccentricity is the \emph{radius} $\radius(G)$.

The notion of \emph{far-apart} pairs of vertices has been introduced in~\cite{SotoGomez2011,Nogues2009} to reduce the number of 4-tuples to consider in the computation of the hyperbolicity (see \Cref{sec:hyperbolicity}).  Roughly, we say that two vertices $u,v\in V$ are \emph{far-apart} if for all $w \in V$ neither
$u$ lies on a shortest path from $w$ to $v$, nor $v$ lies on a shortest path from $u$ to $w$. More formally, we have:

\begin{definition}\label{def:far}
In a graph $G = (V, E)$, vertex $u$ is \emph{far} from vertex $v$, or $v$-far, if for any neighbor $w$ of $u$, we have $\dd(v, w) \leq \dd(v,u)$. The pair $u,v$ of vertices is \emph{far-apart} if $u$ is $v$-far and $v$ is $u$-far.
\end{definition}

The number of far-apart pairs in a graph can be orders of magnitude smaller than the total number of pairs. For instance, a $p \times q$ grid has only 2 far-apart pairs. On the other hand, all pairs in a clique graph are far-apart.

The set of all far-apart pairs can be determined in time $\OO(nm)$ in unweighted graphs through breadth-first search (BFS), and the interested reader will find in \Cref{sec:tradeoffs} a discussion on several time and space complexity trade-offs for determining far-apart pairs.
We now present some interesting properties of far vertices.

\begin{lemma}\label{lem:far}
For $v\in V$, vertex $u\in V$ is $v$-far if and only if $u$ is a leaf of all shortest path trees rooted at $v$.
\end{lemma}
\begin{proof}
Clearly, a non-leaf vertex $u$ of a shortest path tree rooted at $v$ has a neighbor at distance $\dd(v,u)+1$ and cannot be $v$-far. Reciprocally, if $u$ is not $v$-far, it has a neighbor $w$ satisfying $\dd(v,w)=\dd(v,u)+1$. Modifying any shortest path tree by setting $u$ as the parent of $w$ yields a valid shortest path tree where $u$ is not a leaf. 
\end{proof}

From \Cref{lem:far}, we also get the following immediate results.

\begin{corollary} \label{cor:ecc}
For each $u\in V$, any $v\in V$ such that $\dd(u, v) = \ecc(u)$ is $u$-far.
\end{corollary} 

\begin{corollary} \label{cor:diameterisfarapart}
For any $u, v \in V$, if $\dd(u,v) = \diam(G)$, then the pair $(u, v)$ is far-apart.
\end{corollary}

In particular, \Cref{cor:ecc} implies that for any vertex $u$, the set of $u$-far vertices is non-empty.

Interestingly, knowing the far vertices of a node $u$ and their distance to $u$, it is possible to scan distant nodes in a sort of backward BFS
as described in \Cref{sec:distfromfarvertices}.

\section{Iterator over far-apart pairs} \label{sec:iterator}

In \Cref{sec:tradeoffs} we present several algorithmic choices that result in different time and space complexity trade-offs for determining the set of far-apart pairs.
Here, we engineer a data structure and algorithms to determine the set of far-apart pairs and return these pairs sorted by decreasing distances. Our objective is to provide an iterator that determines the next pair to yield on the fly, that postpones computations as much as possible, and with an acceptable memory consumption.

\begin{description}
\item[Data structure.]
To store and organize data, we use an array $F$ indexed by distances in range from 1 to $\diam(G)$, so that cell $F^d$ contains data related to far-apart pairs at distance~$d$.
More precisely, $F^d$ is a hash map associating to a vertex $u \in V$ the subset $F_u^d$ of $u$-far vertices at distance $d$ from $u$. The subset $F_u^d$ can be implemented using either a set data structure allowing to answer a membership query in $\OO(1)$ time, or an array of size $|F_u^d|$ whose elements are sorted according to any total ordering of the vertices to enable membership queries in time $\OO(\log_2{|F_u^d|})$.

We assume, as it is the case for most modern programming languages, that it is possible to visit the elements of a hash map in a fixed arbitrary order (e.g., insertion order). We have the same assumption for set data structures.

\item[Initialization.]
Recall that our aim is to iterate over far-apart pairs by non increasing distances and to postpone computation as much as possible. To this end, we initially store in $F$, for each vertex $u$ with eccentricity $\ecc(u)$, an empty set of $u$-far vertices in $F^d$ with $d={\ecc(u)}$. This empty set will serve as an indicator to trigger the effective computation of the set of $u$-far vertices the first time it is needed (recall that by \Cref{cor:ecc}, this set is non-empty).

Clearly, this initialization procedure requires the knowledge of the eccentricities of all vertices. Fortunately, although the determination of the eccentricities has worst-case time complexity in $O(nm)$, practically efficient algorithms have been proposed to perform this task~\cite{DraganHV18,TakesK13}. These algorithms maintain upper and lower bounds on the eccentricity of each vertex and improve these bounds by computing distances from a few well chosen vertices until all gaps are closed. In practice, these algorithms are orders of magnitude faster than a naive algorithm performing a BFS from each vertex.

\item[Filling.]
Each time we compute distances from a vertex $u$ that has not been considered before, the corresponding sets $F_u^d$ for all $d$ can be inserted in $F$. This is done in particular while computing the eccentricities during the initialization of the data structure. 
 
\item[Next.]
We define a function $\texttt{next}(F)$ that yields the next far-apart pair in the ordering. 
Observe that, since $F$ is used to iterate over the far-apart pairs by non increasing distances, when starting to consider far-apart pairs at distance $d$, we know that all pairs at distance $d'> d$ have already been considered and that at initialization, we have added in $F^d$ an empty set for each vertex with eccentricity $d$. Hence, all vertices involved in a far-apart pair at distance $d$ are known. 

Therefore, we can define a function $\texttt{next}(F^d)$ that returns either the next item $(u, F_u^d)$ in the fixed arbitrary ordering on the items stored in the hash map $F^d$, or \texttt{Stop} when all items have been considered.
Similarly, we define a function $\texttt{next}(F_{u}^d)$ that returns either the next vertex in the ordering  defined over the vertices in $F_{u}^d$, or \texttt{Stop} when all vertices have been considered.

For a distance $d$ and a vertex $u$, we repeat calls to $\texttt{next}(F_{u}^d)$ until finding a vertex $w$ such that $u$ is $w$-far (by testing if $u$ is in $F_w^d$), in which case the far-apart pair $(u, w)$ at distance $d$ is returned. If no vertex $w$ is found, we go to the next vertex in $F^d$ through a call to $\texttt{next}(F^d)$.
When all vertices in $F^d$ have been considered, we start considering far-apart pairs at distance $d-1$.

The use of a total ordering on the vertices allows to ensure that a far-apart pair $(u, v)$ is returned only once, i.e., when $u < v$.
During these operations, if we encounter an empty set $F_x^d$, we know from the initialization step that $x$ has never been considered before and that $\ecc(x) = d$. So, we compute distances from $x$ and store the $x$-far vertices in $F$ via the filling step.

\item[Improvements.] Depending on the usage of this data structure, some simple improvements can be done, such as:
\begin{enumerate}

    \item When iterating over the vertices in $F_u^d$, each time a vertex $w\in F_u^d$ is found such that $u$ is not $w$-far, we can remove $w$ from $F_u^d$. This avoids checking twice if a pair is far-apart, and,  at the end of the iterations on $F_u^d$, this set will contain a vertex $w$ only if $u$ is $w$-far.
    Actually, instead of removing elements from $F_u^d$, it is safer to insert the vertices $w$ such that $u$ is $w$-far into a temporary array (or set) $T$, and to replace $F_u^d$ by $T$ when $\texttt{next}(F_{u}^d)$ returns \texttt{Stop}. \label{improvement:1}
    We will use this improvement in our algorithm for computing hyperbolicity.
    
    \item When all vertices in $F^d$ have been considered, one can delete $F^d$ to reduce the memory consumption if appropriate with the usage. Similarly, one can avoid storing $F^d$ if only far-apart pairs at distance $d'>d$ are requested, as is the case for instance to enumerate the diameters of a graph or when computing hyperbolicity when a lower-bound has been found. \label{improvement:2}
\end{enumerate}

\end{description}

The time complexity for iterating over all far-apart pairs sorted by non increasing distance is in $\OO(nm)$ when using sets for $F_u^d$ (resp., $\OO(nm + n^2\log{n})$ when using arrays). Indeed, the algorithm computes BFS distances from each vertex of the graph (some during initialization, and others during queries). Furthermore, the query time to report all far-apart pairs involving a vertex $u$ is in $\OO(n)$ since $|\cup_{d=1}^{\ecc(u)} F_u^d| = \OO(n)$ and checking if $u\in F_w^d$ requires time $\OO(1)$ (resp., $\OO(\log{n}))$. The sorting operation over the far-apart pairs is implicit and thus adds no extra cost. The space complexity of the data structure is in $\OO(n^2)$. However, we will see in \Cref{sec:experimental} that it is much smaller in practice.

\section{Use Case: Gromov Hyperbolicity} \label{sec:hyperbolicity}
We now present an interesting use case for our far-apart pair iterator with the computation of the Gromov hyperbolicity of a graph~\cite{Gromov1987}.

\subsection{Definitions}
A metric space $(V,\dd_G)$ is a tree metric if there exists a distance-preserving mapping from $V$ to the nodes of an edge-weighted tree. If so, the graph $G$ is said to be $0$-hyperbolic because it satisfies the 4-points condition bellow with parameter $\delta=0$. Introducing some slack $\delta$ allows to measure how much the metric of a graph deviates from that of a tree. This slack $\delta$ is called the \emph{hyperbolicity} of the graph:

\begin{definition}[\textbf{$4$-points Condition,~\cite{Gromov1987}}]
\label{def:hyperbolic}
  Let $G$ be a connected graph.
  For every $4$-tuple $u,v,x,y$ of vertices of $G$, we define $\delta(u,v,x,y)$ as half of the difference between the two largest sums among
  $S_1=\dd(u,v)+\dd(x,y)$, $S_2 = \dd(u,x)+\dd(v,y)$, and $S_3 = \dd(u,y)+\dd(v,x)$.

  The hyperbolicity of $G$, denoted by $\delta(G)$, is equal to $\max_{u,v,x,y\in V(G)} \delta(u,v,x,y)$.
Moreover, we say that $G$ is \emph{$\delta$-hyperbolic} whenever $\delta \geq \delta(G)$.
\end{definition}

Other characterizations exist for $0$-hyperbolic graphs and yield other definitions for the hyperbolicity $\delta$ of a graph that differ only 
by a small constant factor~\cite{Bermudo2013,LaHarpe1990,Gromov1987}.

From the $4$-points condition above, it is straightforward to compute graph hyperbolicity in $\Theta(n^4)$-time.
In theory, it can be decreased to $\OO(n^{3.69})$ by using a clever $(\max,\min)$-matrix product~\cite{FournierIV15}; however, in practice, the best-known algorithms still run in $\OO(n^4)$-time~\cite{BorassiCCM15,CohenCL15}.
Graphs with small hyperbolicity can be recognized faster. In particular, $0$-hyperbolic graphs coincide with \emph{block graphs}, that are graphs whose biconnected components are complete subgraphs~\cite{Bandelt1986,Howorka1979}. Hence, deciding whether a graph is $0$-hyperbolic can be done in time $\OO(n+m)$.
The latter characterization of $0$-hyperbolic graphs follows from a more general result stating that the hyperbolicity of a graph is the maximum hyperbolicity of its biconnected components (see e.g. \cite{CohenCL15} for a proof). 
More recently, it has been proven that the recognition of $\frac{1}{2}$-hyperbolic graphs is computationally equivalent to deciding whether
there exists a chordless cycle of length $4$ in a graph~\cite{CoudertD14}. The latter problem can be solved in deterministic $\OO(n^{3.26})$-time~\cite{LeGall2012} and in randomized $\OO(n^{2.373})$-time~\cite{VassilevskaWilliams2015} by using fast matrix multiplication.

Several pre-processing methods for reducing the size of the input graph have been proposed. In particular, \cite{SotoGomez2011} proved that the hyperbolicity of $G$ is equal to the maximum of the hyperbolicity of the graphs resulting from both a \emph{modular}~\cite{Gallai1967,Habib2010} or a \emph{split}~\cite{Cunningham1980,Cunningham1982} decomposition of $G$. These decompositions can be computed in linear time~\cite{CharbitMR12}.  
Algorithms with time complexity in $\OO(\mw(G)^{3}\cdot n+m)$ and $\OO(\sw(G)^{3}\cdot n+m)$ when parameterized  respectively by the  \emph{modular width} $\mw(G)$ and the \emph{split width} $\sw(G)$ have been proposed in~\cite{CoudertDP19}.
Moreover, \cite{CohenCDL17} shows how to use modified versions  of the atoms of a decomposition of $G$ by clique-minimal separators~\cite{Tarjan1985,Berry2010b,CoudertD18}. This decomposition can be obtained in time $\OO(nm)$.

\subsection{Previous algorithm}
In this section, we recall the algorithm proposed in~\cite{BorassiCCM15} for computing  hyperbolicity. This algorithm improves upon the algorithm proposed in~\cite{CohenCL15} by adding pruning techniques to further reduce the number of 4-tuples to consider.
To the best of our knowledge, the algorithms proposed in~\cite{BorassiCCM15,CohenCL15} are the only algorithms enabling to compute the exact hyperbolicity of graphs with up to $50000$ nodes. These algorithms are based on the following lemmas.

\begin{lemma}[\cite{SotoGomez2011,CohenCL15}]\label{lem:certificate}
  Let $G$ be a connected graph.  There exist two far-apart pairs $(u,v)$ and $(x,y)$ satisfying $\delta(u,v,x,y) = \delta(G)$.
\end{lemma}

\begin{lemma}[\cite{CohenCL15}]\label{lem:shape}
 For every $4$-tuple $u,v,x,y$ of vertices of a connected graph $G$, we have $\delta(u,v,x,y)\leq \min_{a,b\in\set{u,v,x,y}}\dd(a,b)$. Furthermore, if $S_1=\dd(u,v)+\dd(x,y)$ is the largest of the sums defined in \Cref{def:hyperbolic} (which can be assumed w.l.o.g.), we have $\delta(u,v,x,y)\leq \frac{1}{2}\min\set{\dd(u,v),\dd(x,y)}$.
\end{lemma}

For the sake of completeness, we quickly give the proof of Lemma~\ref{lem:shape}. Without loss of generality the second largest sum is $S_2 = d(u,x)+d(v,y)$. By triangle inequality, we have $\dd(u,v)\le \dd(u,x)+\dd(x,y)+\dd(y,v) = S_2 + \dd(x,y)$, yielding $S_1-S_2\le 2\dd(x,y)$. We can similarly obtain $S_1-S_2\le 2\dd(u,v)$.

The key idea of the algorithms of~\cite{BorassiCCM15,CohenCL15} is to visit the most promising 4-tuples first, that is, those made of pairs of far-apart vertices at largest distance, and to stop computation as soon as the bounds of \Cref{lem:shape} are reached. These algorithms thus need to iterate over far-apart pairs ordered by decreasing distances.

More precisely, \Cref{alg:hyp} bellow (see also \cite{BorassiCCM15}) iterates first over the far-apart pairs sorted by non increasing distances. Then, given the $i$-th far-apart pair $(x_i, y_i)$, it iterates over the previous far-apart pairs $(v, w)$, such that $\dd(v, w) \geq \dd(x_i, y_i)$ in order to consider quadruples $(v,w,x_i,y_i)$ such that $S_1=d(v,w)+d(x_i,y_i)$ is the largest sum. 
Note that such pairs $(v,w)$ have been considered previously and satisfy $(v,w)=(x_j,y_j)$ for some $j<i$. 
This ensures by \Cref{lem:shape} that as soon as $\dd(x_i, y_i) \leq 2\delta_L$, where $\delta_L$ is the current best solution, no further improvements can be done. So, the hyperbolicity $\delta_L$ of the graph is then returned in Line~5.
To iterate over the pairs $(v,w)$ such that $\dd(v, w) \geq \dd(x_i, y_i)$, the algorithm maintains the \emph{mates} of each vertex.
Vertex $w$ is a mate of $v$ if $(v,w)$ is a far-aprt pair satisfying $\dd(v,w)\geq \dd(x_i,y_i)$. In other words, $(v,w)$ is a far-apart pair previously considered for some $j<i$ such that $(x_j, y_j) = (v, w)$ or $(x_j, y_j) = (v, w)$.  

\SetKwFunction{computeAccVal}{computeAccVal}
To further prune the search space, \cite{BorassiCCM15} introduces the notions of \emph{skippable}, \emph{acceptable} and \emph{valuable} vertices that are computed by \computeAccVal according to the definitions given bellow. \Cref{alg:hyp} can be read before considering the details of this optimization.
Its time complexity is in $\OO(n^4)$ and its space complexity is in $\Theta(n^2)$. Indeed, it not only needs to store the list of far-apart pairs, but also the distance matrix, and the lists of mates.

\begin{algorithm}[t]
\SetKwFunction{computeAccVal}{computeAccVal}
\KwIn{$\FF=(\{x_1,y_1\},\dots,\{x_N,y_N\})$, an ordered list of far-apart pairs. }
\KwIn{$\dd$, the distance matrix. }
$\delta_L \gets 0$ \\
\texttt{mates}[$v$]$ \gets \emptyset$ for each $v$ \\
\For{$i \in [1,N]$}{
    \If{$\dd(x_i,y_i) \leq 2\delta_L$}{\Return $\delta_L$}

    (\texttt{acceptable}, \texttt{valuable}) $\gets$ \computeAccVal{} \\
    \For{$v \in $ \upshape\texttt{valuable}} {
        \For{$w \in $ \upshape\texttt{mates}[$v$]} {
            \If{$w \in$ \upshape\texttt{acceptable}} {
                $\delta_L\gets\max\{\delta_L, \delta(x_i, y_i, v, w)\}$
            }
        }
    }
    add $y_i$ to \texttt{mates}[$x_i$] \\
    add $x_i$ to \texttt{mates}[$y_i$]
}
\Return $\delta_L$
\caption{Algorithm for computing the hyperbolicity proposed in~\cite{BorassiCCM15}}
\label{alg:hyp}
\end{algorithm}

\paragraph{Skippable, acceptable and valuable.}
Given a pair $x,y$ of nodes and a lower bound $\delta_L$ on hyperbolicity, 
\cite{BorassiCCM15} proposes a classification of the nodes to prune as those that cannot lead to any improvement of the lower-bound $\delta_L$ known so far. For instance, a node $v$ such that $\min\{\dd(x,v), \dd(y,v)\} \leq \delta_L$ can be skipped, since by \Cref{lem:shape} we then have $\delta(x, y, v, w) \leq \delta_L$ for any $w$. In \Cref{lem:skippable} we summarize the conditions defining \emph{$(x,y,\delta_L)$-skippable} nodes.

\begin{lemma}[\cite{BorassiCCM15}]\label{lem:skippable}
 A node $v$ is $(x,y,\delta_L)$-skippable if it satisfies any of the following conditions:
 \begin{enumerate}
    \item $v$ does not belong to any far-apart pair considered before $(x, y)$ (Lemma~5 in~\cite{BorassiCCM15});
    \item $\min\{\dd(x,v), \dd(y,v)\} \leq \delta_L$ (\Cref{lem:shape});
    \item $2\ecc(v) - \dd(x, v) - \dd(y, v) < 4\delta_L + 2 - \dd(x, y)$ (Lemma~8 in~\cite{BorassiCCM15});
    \item $\ecc(v)+\dd(x,y)-3\delta_L-\frac{3}{2} < \max\{\dd(x,v), \dd(y,v)\}$ (Lemma~9 in~\cite{BorassiCCM15}).
 \end{enumerate}
\end{lemma}

A node that does not satisfy any condition of \Cref{lem:skippable} is defined as \emph{$(x,y,\delta_L)$-acceptable}, and so it must be considered. This class is further refined in~\cite{BorassiCCM15} with the subset of \emph{$c$-valuable} vertices, where $c$ is any fixed node (a good choice is a node with small eccentricity or centrality) as specified in the following Lemma. 

\begin{lemma}[\cite{BorassiCCM15}]\label{lem:valuable}
Let $c$ be any fixed node. A $(x,y,\delta_L)$-acceptable node $v$ is $c$-valuable if
$2\dd(c, v) - 2\delta_L > \dd(x, v) + \dd(y, v) - \dd(x, y)$.
\end{lemma}

In \Cref{alg:hyp}, the 4-tuples considered with far-apart pair $(x_i, y_i)$ are such that $v$ is $c$-valuable,  $w$ is $(x,y,\delta_L)$-acceptable and $(v, w)$ is a far-apart pair seen previously.
Overall, the classification of the nodes is done in overall time $\OO(n)$ for a given pair $(x_i, y_i)$ and lower bound $\delta_L$. The experiments reported in~\cite{BorassiCCM15} show that this classification leads to a significant reduction of the number of considered 4-tuples as well as computation time.

\subsection{Hub labeling}

A main bottleneck of \Cref{alg:hyp} comes from the $\Theta(n^2)$ memory usage. This can be alleviated by using hub labeling~\cite{Delling2014}, a technique that allows to encode distances in a graph. The technique is also called two-hop labeling~\cite{CohenHKZ02}. It appears to give a very efficient space-time tradeoff in practice. We tried to use it as a replacement of the distance matrix in \Cref{alg:hyp}. It perfectly fits in memory with all practical graphs we could test. However, computing all distances from a given vertex is orders of magnitude slower than performing a BFS from that vertex. As the technique appeared as an inefficient way to extend Algorithm~\ref{alg:hyp}, we abandoned it.

\subsection{Our algorithm}

To improve upon \Cref{alg:hyp}, and in particular to reduce the memory usage in practice, we do the following.
\begin{itemize}
    \item We use the far-apart pair iterator presented in \Cref{sec:iterator} to avoid the pre-computation and storage of the list of far-apart pairs. The use of Improvement~\ref{improvement:1} of \Cref{sec:iterator} ensures that at the end of the visit of $F_u^d$, it contains only the vertices $w$ such that $u$ is $w$-far, and so the pair $(v, w)$ is far-apart.
    
    \item We design a function $\texttt{mates}(F, v, d)$ to iterate over the far-apart pairs at distance $d \geq \dd(x, y)$ that involve $v$ and that have previously been reported by $\texttt{next}(F)$.
    
    When $\dd(x,y)< d \leq \diam(G)$, this function simply yields vertices from $F_v^d$ since the order of operations of the algorithm when using Improvement~\ref{improvement:1} ensures that $F_v^d$ contains only the vertices forming far-apart pairs with $v$.
    
    When $d = \dd(x, y)$, we have to ensure that a $v$-far vertex $w$ is yielded if and only if the pair $(v, w)$ has previously been reported by a call to $\texttt{next}(F)$ (so the pair $(v, w)$ is far-apart).
    To do so, we modify the far-apart pairs iterator and its $\texttt{next}(F)$ function as follows.
    We use an extra hash map $T$ (initially empty) associating to a vertex $u$ the subset of $u$-far vertices at distance $d$ from $u$ that have previously been reported by $\texttt{next}(F)$.
    Then, when $\texttt{next}(F)$ yields a far-apart pair $(x, y)$, we store $y$ in $T_x$ and $x$ in $T_y$. 
    This way, when $d = \dd(x, y)$, function $\texttt{mates}(F, v, d)$ simply has to yield vertices from $T_v$.
    Finally, as soon as function $\texttt{next}(F)$ starts reporting pairs at distance $d-1$, we exchange hash maps $T$ and $F^d$ (alternative implementation of Improvement~\ref{improvement:1}), and proceed with a cleared hash map $T$.

    \item Instead of giving the distance matrix as input to the algorithm, we compute for each pair $(x, y)$ the BFS distances from $x$ and $y$ before the call to \texttt{computeAccVal()}. Since the distance $\dd(v, w)$ is obtained while extracting the mates of $v$ from the far-apart pair iterator, we get all the needed distances to compute $\delta(x, y, v, w)$.
    
    Although these repeated computations of BFS distances have no impact on the overall time complexity of the algorithm, which remains in $\OO(n^4)$, they represent a significant computation time in practice. To reduce the impact on the overall computation time, we propose Optimisations 2 (\Cref{sec:cache}) and 3 (\Cref{sec:prunedBFS}) below.
\end{itemize}

\begin{algorithm}[t]
\SetKwFunction{computeAccVal}{computeAccVal}
\SetKwFunction{FA}{F}
\SetKwFunction{hasnext}{has\_next}
\SetKwFunction{getnext}{next}
\SetKwFunction{prunedBFS}{prunedBFS}
\SetKwFunction{lowerBoundHeuristic}{lowerBoundHeuristic}
\SetKwFunction{mates}{mates}
\KwIn{$G = (V, E)$}
Initialize the far-apart pair iterator $F$ \\
$\delta_L \gets \lowerBoundHeuristic(G)$ \\
\While{$\hasnext(F)$}{
    $(x, y) \gets  \getnext(F)$ \tcp*[f]{provides $\dd(x, y)$}\\
    \If{$\dd(x, y) \leq 2\delta_L$}{\Return $\delta_L$} 
    $\dd_x \gets c_{x,y,\delta_L}$-$\prunedBFS(x)$ \label{alg:hyp2:cache1}\\ 
    $\dd_y \gets c_{y,x,\delta_L}$-$\prunedBFS(y)$ \label{alg:hyp2:cache2}\\ 
	(\texttt{acceptable}, \texttt{valuable}) $\gets$ \computeAccVal{} \\
    \For{$v \in $ \upshape\texttt{valuable}} {
        \For(\tcp*[f]{provides $\dd(v, w)$}){$w \in \mates(F, v, d)$} {
            \If{$w \in$ \upshape\texttt{acceptable}} {
                $\delta_L\gets\max\{\delta_L, \delta(x, y, v, w)\}$ 
            }
        }
    }
}
\Return $\delta_L$
\caption{New algorithm for computing the hyperbolicity}
\label{alg:hyp2}
\end{algorithm}

See \Cref{alg:hyp2} for the overall presentation of our algorithm.
In the following, we present some optimisations aiming at reducing the computation time.

\subsubsection{Optimisation 1: Lower bound initialisation}
\label{sec:lower_bound_init}

The  technique of acceptable and valuable nodes becomes more efficient when $\delta_L$ is larger as Inequalities~3 and~4 of \Cref{lem:skippable} and the inequality of \Cref{lem:valuable} become stricter. For that reason, we first use the heuristic described in \cite{CohenCL15} to set an initial value to $\delta_L$. It is referenced as lowerBoundHeuristic in \Cref{alg:hyp2}. 

\subsubsection{Optimisation 2: Cache of BFSs}
\label{sec:cache}

We use a cache of BFSs to avoid recomputing a BFS that has recently been computed. This cache has a bounded capacity (e.g., 1000 BFSs).
Observe that even a cache of 2 BFSs is beneficial as function $\texttt{next}(F)$ reports successively all far-apart pairs at distance $d$ involving a vertex $x$ and such that $x < y$.

\subsubsection{Optimisation 3: Pruned BFS for searching acceptable nodes}
\label{sec:prunedBFS}

When considering a pair $(x,y)$, we perform BFS searches from both $x$ and $y$ to obtain distances from $x$ and $y$ and detecting both acceptable and valuable nodes. \Cref{lem:skippable} allows to restrict both searches as follows.

First, we observe that if a node $v$ satisfies $\ecc(v) - \dd(x,v) < 3\delta_L - \frac{3}{2} - \dd(x,y)$, then Lemma~\ref{lem:skippable} applies and $v$ is $(x,y,\delta_L)$-skippable. A $(x,y,\delta_L)$-acceptable node $v$ must thus satisfy:
\begin{equation}\label{eq:prune}
    \ecc(v) - \dd(x,v) \ge c_{x,y,\delta_L}, 
    \quad\mbox{where}\quad 
    c_{x,y,\delta_L}=3\delta_L - \frac{3}{2} - \dd(x,y).
\end{equation}

We then define the \emph{$c_{x,y,\delta_L}$-pruned BFS search} from $x$ as a BFS search from $x$ that visits only nodes satisfying Equation~\eqref{eq:prune}. More precisely, when visiting a node $u$, we enqueue only neighbors $v$ of $u$ that satisfy Equation~\eqref{eq:prune}. Note that $c_{x,y,\delta_L}$ is constant given $x,y$ and $\delta_L$. We can then safely replace the regular BFS from $x$ by a pruned BFS as stated by the following lemma.

\begin{lemma}\label{lem:prunedBFS}
A $c_{x,y,\delta_L}$-pruned BFS search from $x$ visits all $(x,y,\delta_L)$-acceptable nodes.
\end{lemma}

\begin{proof}
The proof follows from two facts. First, any $(x,y,\delta_L)$-acceptable node must satisfy Equation~\eqref{eq:prune}. Second, any node $v\not= x$ satisfying Equation~\eqref{eq:prune} must have some neighbor closer to $x$ that satisfies Equation~\eqref{eq:prune}. This second condition proven bellow allows to easily prove by induction that all nodes satisfying Equation~\eqref{eq:prune} are visited by a $c_{x,y,\delta_L}$-pruned BFS search from $x$.

To prove the above second condition, consider a neighbor $u$ of $v$ at distance $\dd(x,v)-1$ from $x$. By triangle inequality, we have $\ecc(u)\ge \ecc(v)-1$, implying $\ecc(u) - \dd(x,u)\ge \ecc(v) - \dd(x,v)$. As $v$ satisfies Equation~\eqref{eq:prune}, so does $u$.
\end{proof}

Notice that the set $V_{x,y,\delta_L}$ of nodes visited by a $c_{x,y,\delta_L}$-pruned BFS search from $x$ is larger than the set of $(x,y,\delta_L)$-acceptable nodes. Indeed, Conditions 1 to 3 of \Cref{lem:skippable} cannot be used for the search.
Furthermore, remark that the set $V_{x,y,\delta_L}$ depends on the distance $\dd(x,y)$ and not on the precise node $y$. That is,
\begin{lemma} \label{lem:prunedBFS:distance}
For any $z$ such that $\dd(x, y)= \dd(x,z)$, we have $V_{x,z,\delta_L} = V_{x,y,\delta_L}$.
\end{lemma}

The following results are direct consequences of Equation~\eqref{eq:prune} and \Cref{lem:prunedBFS:distance}.
\begin{corollary}\label{cor:prunedBFS:eq}
 For any $z$ such that $\dd(x, y) \geq \dd(x,z)$, we have $ V_{x,z,\delta_L}\subseteq V_{x,y,\delta_L}$.
\end{corollary}

\begin{corollary}\label{cor:prunedBFS:sub}
 For any $\delta > \delta_L$, we have $V_{x,y,\delta}\subseteq V_{x,y,\delta_L}$.
\end{corollary}

\Cref{lem:prunedBFS:distance,cor:prunedBFS:eq,cor:prunedBFS:sub} enable the use of our pruned BFS in combination with a cache of BFSs. Indeed, during the execution of the algorithm both the considered distance $\dd(x, y)$ decreases and the lower bound $\delta_L$ increases. Hence, a cached $c_{x,y,\delta_L}$-pruned BFS remains valid for future use.
Hence, for Line~\ref{alg:hyp2:cache1} of \Cref{alg:hyp2}, we first check if a BFS from $x$ is in the cache, and if so we retrieve it. Otherwise, we perform a $c_{x,y,\delta_L}$-pruned BFS search from $x$ and add it to the cache.
We proceed similarly for $y$.

Observe also that to determine the sets of $(x,y,\delta_L)$-acceptable and $c$-valuable vertices, it suffices to consider vertices that have been visited by the $c_{x,y,\delta_L}$-pruned BFS search, or by the $c_{y,x,\delta_L}$-pruned BFS search if this set is smaller.

\section{Experimental results}
\label{sec:experimental}

In this section we conduct experiments to test the performance of our new algorithm in comparison to the previous state-of-the-art one. 
We additionally test the impact that each proposed optimization has on the overall performance. To gain further insights into our main tool, the algorithm for efficiently enumerating far-apart pairs, we also conduct experiments to analyze the performance and the number of far-apart pairs that graphs exhibit.

\subsection{Implementation notes}
\label{sec:implementation_notes}

We have implemented all the algorithms in C++ and our code is available~\cite{gitlab}.
In this section, we discuss some implementation choices.

We have implemented a cache of BFSs with bounded capacity $\kappa$ that additionally maintains the age information of the data it stores.
We use a counter $\tau$, initialized to 0, that is increased by one each time a BFS that is already in the cache is accessed, or a BFS that is not in the cache is added.
We associate to a cached BFS from a vertex $x$ an age information $a_x$, initialized to the value of $\tau$ at insertion time. Then, we set $a_x$ to the current value of $\tau$ each time the BFS from $x$ is accessed. Hence, the last accessed BFS is such that $a_x = \tau$.
The use of a hash map associating to a vertex the corresponding BFS and age information enables to decide in  time $\OO(1)$ if a BFS from $x$ is in the cache, and if so to return a pointer on the corresponding data. 
Updating the age information is also done in time $\OO(1)$.
The insertion of a BFS in the cache takes time $\OO(1)$ if the cache is not full, and time $\OO(\kappa)$ as soon as it has reached its maximum capacity. Indeed, the insertion of a BFS when the cache is full requires to remove first the BFS with largest age information, that is, the one of the vertex $x$ maximizing $\tau - a_x$, and so with smallest $a_x$. 
Note that for $\kappa$ much smaller than $n$, the time required for managing the cache is negligible with respect to the time required for a BFS. 
Observe that this cache will be accessed $\OO(n^2)$ times by \Cref{alg:hyp2}, and more precisely at most twice the number of far-apart pairs at distance $2\delta(G)$ or more.

\subsection{Data \& Hardware}

We test
graphs from the BioGRID interaction database (BG-*)~\cite{BIOGRID};
a protein interactions network (dip20170205)~\cite{DIP}; 
and graphs of the autonomous systems from the Internet (CAIDA\_as\_* and DIMES\_*)~\cite{CAIDA,DIMES}.
We also test social networks (Epinions, Hollywood, Slashdot, Twitter),
co-author graphs (ca-*, dblp),
computer networks (Gnutella, Skitter), 
web graphs (NotreDame), 
road networks (oregon2, FLA-t), 
a 3D triangular mesh (buddha), 
and grid-like graphs from VLSI applications (alue7065) 
and from computer games (FrozenSea).
The data is available from \url{snap.stanford.edu}, \url{webgraph.di.unimi.it}, \url{www.dis.uniroma1.it/challenge9}, \url{graphics.stanford.edu}, \url{steinlib.zib.de}, and \url{movingai.com}. Furthermore, we test synthetic inputs: grid300-10 and grid500-10 are square grids with respective sizes $301\times 301$, and $501\times 501$ where 10\% of the edges were randomly deleted.
Each graph is taken as an undirected unweighted graph and we consider only its largest bi-connected component (available from~\cite{gitlab}).
See \Cref{tab:graphs} for the characteristics of these graphs.

We used a computer equipped with two Intel Xeon Gold 6240 CPUs operating at 2.6GHz and 192G RAM to run our experiments. Note that our code uses a single thread.

\begin{table}
\caption{Some graph parameters of all the graphs that we use in our experiments. Note that for all graphs, we extracted the largest biconnected component and restrict to this subgraph in our experiments.}
\label{tab:graphs}
\centering
\begin{tabular}{|l|ccccc|}
\hline
\textbf{Graph} & \textbf{\#nodes} & \textbf{\#edges} & \textbf{radius} & \textbf{diameter} & \textbf{mean ecc.} \\
\hline
\hline
BG-MV-Physical & 9851 & 45558 & 11 & 22 & 14.45 \\
BG-S-Affinity\_Capture-MS & 17793 & 174210 & 6 & 12 & 9.02 \\
BG-S-Affinity\_Capture-RNA & 3339 & 10408 & 3 & 5 & 4.00 \\
BG-S-Affinity\_Capture-Western & 9971 & 44331 & 10 & 20 & 12.60 \\
BG-S-Biochemical\_Activity & 2944 & 10444 & 7 & 13 & 9.20 \\
BG-S-Dosage\_Rescue & 1521 & 4143 & 9 & 17 & 12.21 \\
BG-S-Synthetic\_Growth\_Defect & 3013 & 21341 & 4 & 8 & 5.72 \\
BG-S-Synthetic\_Lethality & 2258 & 12187 & 4 & 7 & 5.45 \\
dip20170205 & 13969 & 60621 & 10 & 17 & 12.95 \\
\hline
CAIDA\_as\_20000102 & 4009 & 10101 & 4 & 8 & 5.62 \\
CAIDA\_as\_20040105 & 10424 & 27061 & 4 & 8 & 5.98 \\
CAIDA\_as\_20050905 & 12957 & 33541 & 5 & 8 & 6.11 \\
CAIDA\_as\_20110116 & 23214 & 89783 & 4 & 8 & 6.04 \\
CAIDA\_as\_20120101 & 25614 & 109180 & 4 & 8 & 6.10 \\
CAIDA\_as\_20130101 & 27454 & 124672 & 5 & 10 & 7.21 \\
CAIDA\_as\_20131101 & 29432 & 143000 & 5 & 9 & 6.46 \\
DIMES\_201012 & 18764 & 84851 & 4 & 7 & 5.16 \\
DIMES\_201204 & 16907 & 66489 & 4 & 7 & 5.07 \\
\hline
p2p-Gnutella09 & 5606 & 23510 & 5 & 8 & 6.31 \\
gnutella31-d & 33812 & 119127 & 6 & 9 & 7.41 \\
notreDame-d & 134958 & 833732 & 18 & 36 & 20.99 \\
\hline
ca-CondMat & 17234 & 84595 & 6 & 12 & 8.44 \\
ca-HepPh & 9025 & 114046 & 6 & 11 & 7.83 \\
ca-HepTh & 5898 & 20983 & 7 & 11 & 8.63 \\
com-dblp.ungraph & 211409 & 883570 & 8 & 15 & 10.92 \\
dblp-2010 & 140610 & 572873 & 10 & 17 & 11.99 \\
email-Enron & 20416 & 163257 & 5 & 9 & 6.54 \\
epinions1-d & 36111 & 365253 & 5 & 9 & 6.36 \\
facebook\_combined & 3698 & 85963 & 4 & 6 & 5.26 \\
loc-brightkite & 33187 & 188577 & 6 & 11 & 7.95 \\
loc-gowalla\_edges & 137519 & 887929 & 6 & 11 & 7.91 \\
slashdot0902-d & 51528 & 473218 & 5 & 8 & 6.04 \\
\hline
oregon2\_010331 & 7602 & 27870 & 4 & 8 & 5.89 \\
t.FLA-w & 691175 & 941893 & 890 & 1780 & 1378.52 \\
\hline
buddha-w & 543652 & 1631574 & 244 & 487 & 360.44 \\
froz-w & 749520 & 2895228 & 812 & 1451 & 1130.38 \\
grid300-10 & 90211 & 162152 & 300 & 600 & 450.50 \\
grid500-10 & 250041 & 449831 & 500 & 1000 & 750.49 \\
z-alue7065 & 34040 & 54835 & 213 & 426 & 319.43 \\
\hline
\end{tabular}
\end{table}

\subsection{Parameter choice}
As there are instances which do not terminate in reasonable time or need unreasonable amounts of memory, we cap both resources at a fixed value to obtain a clear picture. More precisely, for each graph and each experiment, we kill the process as soon as it takes more than 6 hours or uses more than 192 GB of memory. We use the $\skull$ symbol to indicate a killed process and, if applicable, put this symbol in the column (i.e., time or memory) that caused the process to be killed.
Furthermore, for the cache described in Section~\ref{sec:implementation_notes}, we use a size of 1000, unless mentioned otherwise.

\subsection{Comparison to previous work}

To evaluate the performance improvement over previous approaches, we compare to the practical state-of-the-art algorithm, which was given in~\cite{BorassiCCM15}.
To this end, we measure the single-threaded computation time, the memory, and the best upper and lower bound found for our new algorithm as well as the one by Borassi et al.~\cite{BorassiCCM15}. The results of this experiment are shown in \Cref{tab:borassi_vs_ours}.

\begin{table}
\caption{Comparing the memory and time consumption for our algorithm and \cite{BorassiCCM15}. If there is a timeout, we state the best lower and upper bounds on the hyperbolicity that we obtained.}
\label{tab:borassi_vs_ours}
\centering
  \resizebox{\linewidth}{!}{
    \renewcommand{\arraystretch}{1.1}
\begin{tabular}{|l|ccc|ccc|}
\hline
{\bfseries Graph}	&	\multicolumn{3}{c|}{{\bfseries Borassi et al. \cite{BorassiCCM15}}}	&	\multicolumn{3}{c|}{{\bfseries Our Algorithm}}\\\hline
					&	{\bfseries time (s)}	&	{\bfseries memory}	& \textbf{hyperb.} & {\bfseries time (s)}	&	{\bfseries memory} & \textbf{hyperb.} \\
\hline
\hline
BG-MV-Physical & 10429 & 1.03 GB & 4.5 & 6725 & 166.33 MB & 4.5 \\
BG-S-Affinity\_Capture-MS & $\skull$ & -- & [2.0, 4.0] & $\skull$ & -- & [2.0, 4.0] \\
BG-S-Affinity\_Capture-RNA & 0.77 & 157.18 MB & 2.0 & 0.67 & 57.29 MB & 2.0 \\
BG-S-Affinity\_Capture-Western & 7827 & 1.05 GB & 4.0 & 5255 & 154.66 MB & 4.0 \\
BG-S-Biochemical\_Activity & 8.45 & 104.90 MB & 3.0 & 16.32 & 46.16 MB & 3.0 \\
BG-S-Dosage\_Rescue & 11.97 & 30.54 MB & 4.0 & 14.31 & 24.43 MB & 4.0 \\
BG-S-Synthetic\_Growth\_Defect & 1.73 & 108.76 MB & 2.0 & 2.93 & 43.82 MB & 2.0 \\
BG-S-Synthetic\_Lethality & 1.11 & 77.09 MB & 2.0 & 1.75 & 33.40 MB & 2.0 \\
dip20170205 & $\skull$ & -- & [4.5, 5.0] & 18318 & 339.30 MB & 4.5 \\
\hline
CAIDA\_as\_20000102 & 1.28 & 260.87 MB & 2.5 & 1.25 & 56.40 MB & 2.5 \\
CAIDA\_as\_20040105 & 18.43 & 1.90 GB & 2.5 & 25.87 & 166.40 MB & 2.5 \\
CAIDA\_as\_20050905 & 16.9 & 2.37 GB & 3.0 & 21.77 & 203.72 MB & 3.0 \\
CAIDA\_as\_20110116 & 13806 & 8.41 GB & 2.0 & 13491 & 556.72 MB & 2.0 \\
CAIDA\_as\_20120101 & $\skull$ & -- & [2.0, 2.5] & $\skull$ & -- & [2.0, 2.5] \\
CAIDA\_as\_20130101 & 2961.95 & 10.09 GB & 2.5 & 3535.82 & 1.47 GB & 2.5 \\
CAIDA\_as\_20131101 & $\skull$ & -- & [2.0, 2.5] & 16108 & 593.03 MB & 2.5 \\
DIMES\_201012 & 465.76 & 4.86 GB & 2.0 & 6043 & 482.41 MB & 2.0 \\
DIMES\_201204 & 66.35 & 4.34 GB & 2.0 & 56.17 & 304.48 MB & 2.0 \\
\hline
p2p-Gnutella09 & 2.9 & 381.64 MB & 3.0 & 2.35 & 98.38 MB & 3.0 \\
gnutella31-d & 139.54 & 13.13 GB & 3.5 & 109.49 & 1.51 GB & 3.5 \\
notreDame-d & -- & $\skull$ & -- & 4514 & 53.02 GB & 8.0 \\
\hline
ca-CondMat & 112.76 & 3.38 GB & 3.5 & 172.94 & 281.18 MB & 3.5 \\
ca-HepPh & 36.97 & 1.17 GB & 3.0 & 86.35 & 152.86 MB & 3.0 \\
ca-HepTh & 4.02 & 407.78 MB & 4.0 & 2.99 & 91.54 MB & 4.0 \\
com-dblp.ungraph & -- & $\skull$ & -- & 5449 & 14.20 GB & 5.0 \\
dblp-2010 & -- & $\skull$ & -- & 6904 & 7.51 GB & 5.5 \\
email-Enron & 754.13 & 5.36 GB & 2.5 & 897.72 & 404.38 MB & 2.5 \\
epinions1-d & 696.73 & 18.60 GB & 2.5 & 2331.33 & 759.20 MB & 2.5 \\
facebook\_combined & 6919 & 248.00 MB & 1.5 & 4551 & 136.65 MB & 1.5 \\
loc-brightkite & 910.39 & 12.81 GB & 3.0 & 910.12 & 812.98 MB & 3.0 \\
loc-gowalla\_edges & -- & $\skull$ & -- & 14478 & 4.04 GB & 3.5 \\
slashdot0902-d & 9560 & 37.54 GB & 2.5 & 4718 & 1.34 GB & 2.5 \\
\hline
oregon2\_010331 & 72.99 & 980.82 MB & 2.0 & 65.13 & 133.76 MB & 2.0 \\
t.FLA-w & -- & $\skull$ & -- & $\skull$ & -- & [81.0, 835.5] \\
\hline
buddha-w & -- & $\skull$ & -- & $\skull$ & -- & [93.0, 221.0] \\
froz-w & -- & $\skull$ & -- & $\skull$ & -- & [367.5, 633.5] \\
grid300-10 & -- & $\skull$ & -- & 10.13 & 1.08 GB & 280.0 \\
grid500-10 & -- & $\skull$ & -- & 99.64 & 2.99 GB & 463.0 \\
z-alue7065 & 35.01 & 9.07 GB & 138.0 & 33.47 & 431.18 MB & 138.0 \\
\hline
\end{tabular}
}
\end{table}

There are several interesting observations that one can derive from these experiments. First, the new approach that we present in this paper needs significantly less memory. More precisely, the memory consumption is up to a factor of 28 times lower than for \cite{BorassiCCM15} (see the slashdot0902-d graph), only considering the graphs where both approaches stay within the limits. If we also consider the graphs where \cite{BorassiCCM15} runs out of memory, we use at least a factor 177 less memory (see grid300-10). Mainly due to this significant reduction of memory consumption, we are able to compute the hyperbolicity for graphs which were previously out of reach due to excessive amounts of memory which would have been necessary. However, there are also graphs for which we can compute the hyperbolicity, which hit the timeout limit of 6 hours when using \cite{BorassiCCM15}. Note that in these two cases, the computation roughly takes between 4 and 5 hours, so \cite{BorassiCCM15} takes at least 1 to 2 hours more time to finish on these instances. In general, we note that all graphs for which \cite{BorassiCCM15} computes the hyperbolicity within the time and memory limits, also stay within the limits using our approach. Over all 40 benchmark instances, our new approach computes the hyperbolicity for 8 more graphs than \cite{BorassiCCM15}. If the hyperbolicity cannot be computed within the limits with our new approach, then this is due to running out of time instead of running out of memory. This again shows that we achieve a drastic reduction in memory consumption for computing the hyperbolicity.

While we are sometimes faster than \cite{BorassiCCM15}, one could also assume that sometimes it is the other way around. This is indeed the case, and while for almost all graphs we are slightly faster or slightly slower, there is one instance (namely, DIMES\_201012) where we are a factor of 13 slower than \cite{BorassiCCM15}. We suspect that for this graph, there are many nodes from which we repeatedly perform BFSs to obtain distances.

We specifically want to highlight two grid graphs, grid300-10 and grid500-10, for which our approach only takes 10 seconds and 90 seconds, respectively, while \cite{BorassiCCM15} fails to compute the hyperbolicity due to exceeding the 192 GB memory limit. Note that for the former grid graph, only 1.08 GB of memory is needed for our approach and thus the memory usage is lower by at least a factor of 177. This drastic reduction of memory comes from the highly structured input: a perfect grid only has 2 far-apart pairs, the two pairs of opposing corners. Thus, the hyperbolicity computation only needs to consider these two pairs and computing the shortest paths between all pairs is hugely wasteful. Our approach includes this natural intuition to lazily compute the distances between the pairs of nodes to avoid huge running times and memory consumption in cases like these.

\begin{figure}[htbp]
    \centering
	\includegraphics[width=.8\textwidth]{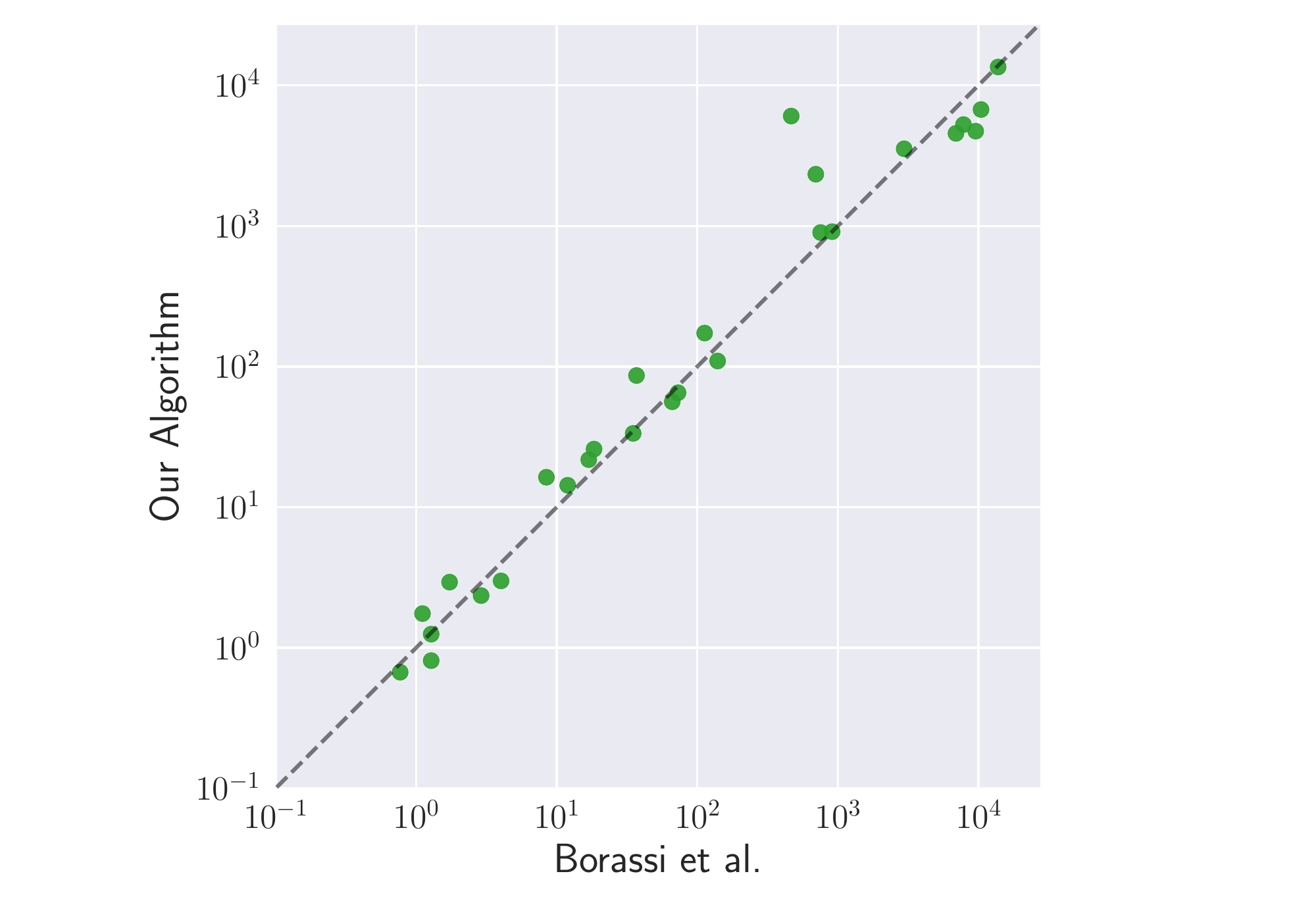}
	\caption{Comparing the \emph{time} in seconds for computing the hyperbolicity on all graphs that finish using both algorithms. The dashed line is the identity, i.e., where both algorithms would take the same time.}
    \label{fig:scatter_time}

	\vspace*{\floatsep}

    \centering
	\includegraphics[width=.8\textwidth]{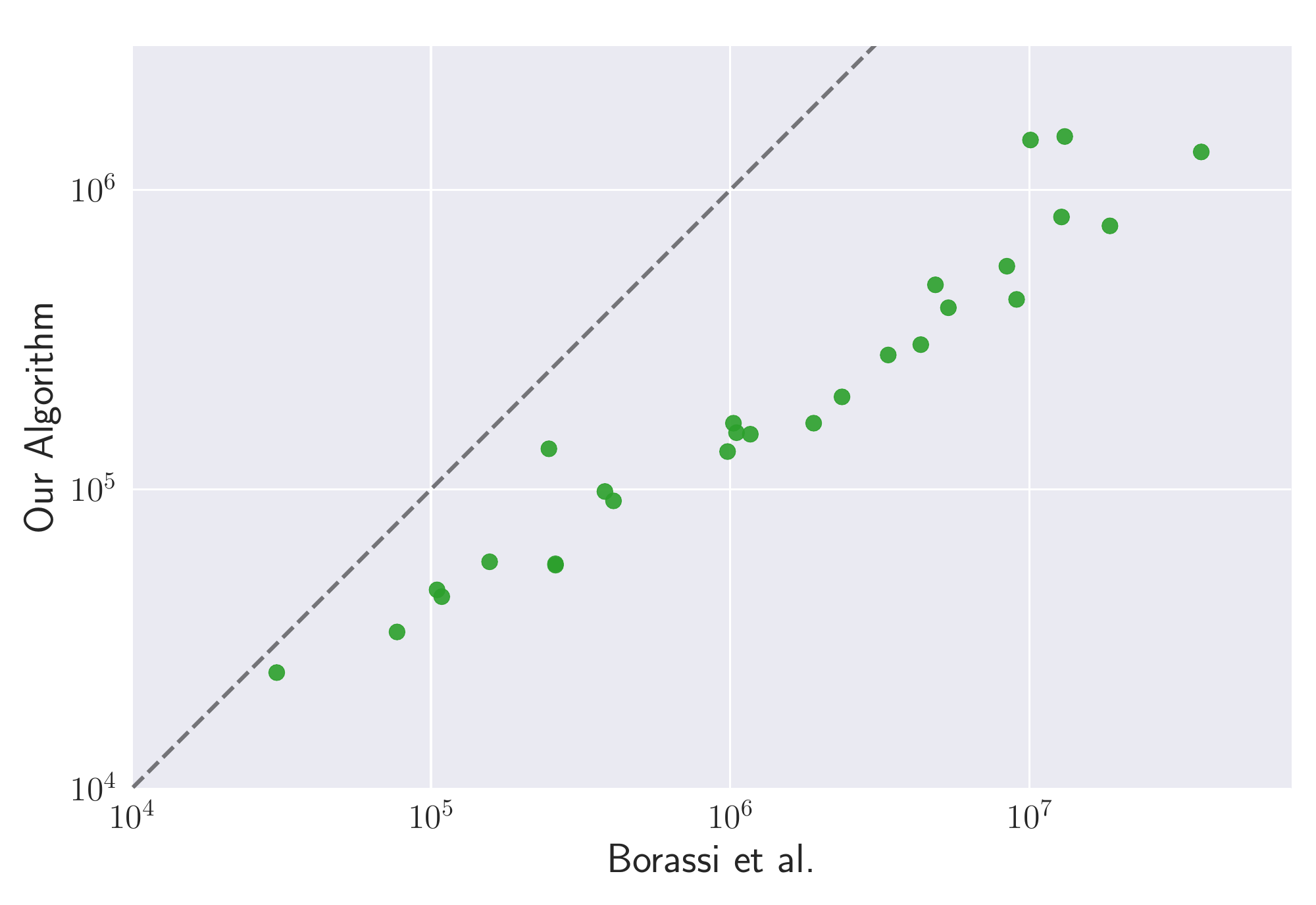}
	\caption{Comparing the \emph{memory} in MB for computing the hyperbolicity on all graphs that finish using both algorithms. The dashed line is the identity, i.e., where both algorithms would use the same amount of memory.}
    \label{fig:memory_time}
\end{figure}

For a better overview, we plotted a comparison of the time and memory usage of both algorithms, see \Cref{fig:scatter_time,fig:memory_time}. In \Cref{fig:scatter_time} we can see that the times indeed are very similar for both approaches except for the single outlier mentioned above. In \Cref{fig:memory_time} we can see the drastic reduction in memory usage. Note again that the graphs included in this figure are only the graphs on which both algorithms terminate within the limits. The graphs where~\cite{BorassiCCM15} exceeds the limits while our approach stays withing the limits are not shown. In particular, the memory reduction by at least a factor of 177 is not shown in this plot. Furthermore, we can see that with increasing memory consumption, also the advantage that our algorithm has over~\cite{BorassiCCM15} with respect to memory usage becomes more and more pronounced.

\subsection{Impact of different optimizations}

Additional to the comparison with the previous state-of-the-art algorithm, we also conduct experiments to obtain insights into the impact of the different optimizations that we used in our algorithm. In particular, we want to know the impact of the lower bound initialization (see \Cref{sec:lower_bound_init}), the cache size (see \Cref{sec:cache}), and the pruning (see \Cref{sec:prunedBFS}). We conduct experiments with all our graphs with the normal cache size of 1000 with no heuristic and no pruning, heuristic but no pruning, pruning but no heuristic, and heuristic and pruning. To gain more insight into the effect of the BFS cache, we also conduct the last experiment  with a cache size of 2, which means that we only store the current two BFSs in memory. See \Cref{tab:prune_heur} for the results of these experiments.

The overall results of these experiments are that there is no significant benefit for the running time when using the heuristically computed initial lower bound. However, the pruning has significant impact depending on the graph: sometimes it does not help much, but in several cases it does decrease the running time by a larger factor -- up to a factor of 5 (see the DIMES\_201204). The cache size also has very different impact on different graphs. It can reduce the time by a factor of 1.4, but it also sometimes increases the running time. However, the positive effects outweigh the negative effects and we thus consider it a worthwhile optimization.

\begin{table}[htbp]
	\caption{Times for computing the hyperbolicity with different optimizations enabled. All entries are in seconds. The columns with \enquote{heur} use the lower bound initialization presented in \Cref{sec:lower_bound_init}, and the columns with \enquote{prune} use the pruning of \Cref{sec:prunedBFS}. The value of $c$ in the second row gives the size of the BFS cache presented in \Cref{sec:cache}.}
\label{tab:prune_heur}
\centering
\begin{tabular}{|l|ccc|cc|}
\hline
{\bfseries Graph}	& \textbf{--} & \textbf{heur} & \textbf{prune} & \multicolumn{2}{c|}{\textbf{heur \& prune}} \\
					& {\boldmath$c=1000$}& {\boldmath$c=1000$}& {\boldmath$c=1000$}& {\boldmath $c=2$} & {\boldmath$c=1000$} \\
\hline
\hline
BG-MV-Physical & 7040 & 7047 & 6745 & 6867 & 6725 \\
BG-S-Affinity\_Capture-MS & $\skull$ & $\skull$ & $\skull$ & $\skull$ & $\skull$ \\
BG-S-Affinity\_Capture-RNA & 0.65 & 0.69 & 0.67 & 0.72 & 0.67 \\
BG-S-Affinity\_Capture-Western & 5464 & 5450 & 5236 & 5297 & 5255 \\
BG-S-Biochemical\_Activity & 16.52 & 16.51 & 16.25 & 20.17 & 16.32 \\
BG-S-Dosage\_Rescue & 14.91 & 14.74 & 14.58 & 16.47 & 14.31 \\
BG-S-Synthetic\_Growth\_Defect & 3.25 & 3.3 & 2.99 & 4.61 & 2.93 \\
BG-S-Synthetic\_Lethality & 2.14 & 2.17 & 1.75 & 2.49 & 1.75 \\
dip20170205 & 19088 & 19144 & 18492 & 18473 & 18318 \\
\hline
CAIDA\_as\_20000102 & 1.24 & 1.26 & 1.23 & 1.72 & 1.25 \\
CAIDA\_as\_20040105 & 39.89 & 40.03 & 25.69 & 34.97 & 25.87 \\
CAIDA\_as\_20050905 & 35.63 & 36.91 & 21.66 & 23.76 & 21.77 \\
CAIDA\_as\_20110116 & 16079 & 16132 & 13700 & 13412 & 13491 \\
CAIDA\_as\_20120101 & $\skull$ & $\skull$ & $\skull$ & $\skull$ & $\skull$ \\
CAIDA\_as\_20130101 & 3292.8 & 3313.78 & 3594.91 & 3848 & 3535.82 \\
CAIDA\_as\_20131101 & 15785 & 15766 & 16036 & 16683 & 16108 \\
DIMES\_201012 & 6144 & 6139 & 6032 & 6574 & 6043 \\
DIMES\_201204 & 294.08 & 292.47 & 56.01 & 68.74 & 56.17 \\
\hline
p2p-Gnutella09 & 2.0 & 2.41 & 1.87 & 2.27 & 2.35 \\
gnutella31-d & 172.25 & 191.22 & 91.08 & 104.68 & 109.49 \\
notreDame-d & 4317 & 4315 & 4472 & 4795 & 4514 \\
\hline
ca-CondMat & 385.97 & 387.39 & 177.16 & 197.9 & 172.94 \\
ca-HepPh & 144.67 & 145.0 & 87.82 & 103.41 & 86.35 \\
ca-HepTh & 3.84 & 4.06 & 2.75 & 4.71 & 2.99 \\
com-dblp.ungraph & $\skull$ & $\skull$ & 5352 & 5713 & 5449 \\
dblp-2010 & 14677 & 14675 & 6610 & 6560 & 6904 \\
email-Enron & 1121.16 & 1108.76 & 906.45 & 1184.8 & 897.72 \\
epinions1-d & 2626.57 & 2607.03 & 2331.39 & 3563.06 & 2331.33 \\
facebook\_combined & 4677 & 4623 & 4552 & 4609 & 4551 \\
loc-brightkite & 1696.58 & 1688.57 & 928.46 & 1108.55 & 910.12 \\
loc-gowalla\_edges & $\skull$ & $\skull$ & 14259 & 13169 & 14478 \\
slashdot0902-d & 4492 & 4463 & 4775 & 6766 & 4718 \\
\hline
oregon2\_010331 & 124.18 & 124.45 & 65.62 & 79.84 & 65.13 \\
t.FLA-w & $\skull$ & $\skull$ & $\skull$ & $\skull$ & $\skull$ \\
\hline
buddha-w & $\skull$ & $\skull$ & $\skull$ & $\skull$ & $\skull$ \\
froz-w & $\skull$ & $\skull$ & $\skull$ & $\skull$ & $\skull$ \\
grid300-10 & 10.2 & 10.36 & 10.13 & 8.7 & 10.13 \\
grid500-10 & 102.21 & 102.31 & 99.84 & 83.08 & 99.64 \\
z-alue7065 & 36.52 & 37.04 & 33.76 & 33.63 & 33.47 \\
\hline
\end{tabular}
\end{table}

\subsection{BFS cache size experiments}
To gain further insights into how the BFS cache we use affects the running time behavior, we run experiments with different cache sizes on selected graphs, see \Cref{fig:cache_size_vs_time}. We again see different behavior on different instances. While the times monotonously decrease with increasing cache size for two of the graphs, for the other two we have increase-decrease and decrease-increase patterns. One reason for this behavior might be the cache size of the processor. In particular, if the graph already fits in the CPU cache (e.g., L1), then computing a BFS is quite fast. Especially for large BFS cache sizes which might push the graph out of the CPU cache and also might reside in a higher level CPU cache, the computation can become slower.

\begin{figure}[tbp]
	\centering
	\includegraphics[width=.8\textwidth]{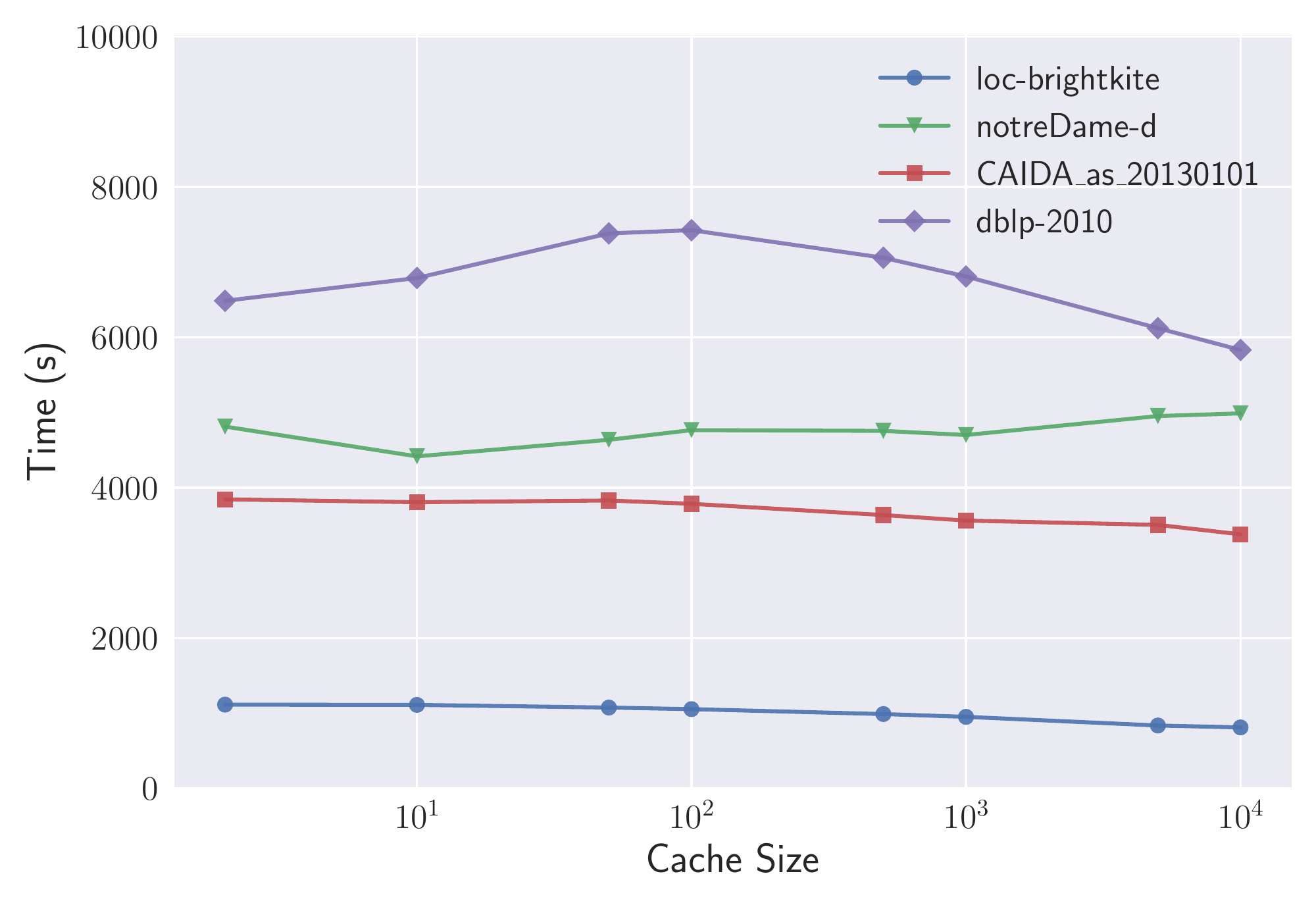}
	\caption{Plot for four graphs showing the running time development depending on the BFS cache size.}
	\label{fig:cache_size_vs_time}
\end{figure}

\subsection{Far-apart pairs iterator experiments}
Finally, we perform experiments to only analyze the behavior of the far-apart pairs iterator. To this end, we let the far-apart pairs iterator run on all our benchmark graphs and measure the time as well as memory consumption. Additionally, we are interested in the number of far-apart pairs that the different instances have as well as how many pairs of an instance are necessary for the hyperbolicity calculation of our new algorithm. As the graphs are of different sizes, we put the number of pairs in the last two columns in relation to the total number of node pairs in the graph. The results of these experiments are shown in \Cref{tab:far_apart_iterator}. For all except four graphs, the far-apart pairs iterator runs through in the given memory and time limits. 
This is one more tractable instance compared to hyperbolicity computation. 
Note, however, that there are graphs for which we can compute the hyperbolicity but the far-apart pairs iterator does not run through within the limits (e.g., dblp-2010). This is explained by the fact that to compute the hyperbolicity, we can stop at some point of iterating through the far-apart pairs and do not have to compute them all. More specifically, we show what percentage of all pairs are necessary to compute the hyperbolicity with our algorithm on this specific instance. Observe that for the dblp-2010 graph, only around 0.27\% of pairs are relevant for the hyperbolicity computation. In general, most instances only have a single-digit percentage or less of relevant pairs for the hyperbolicity computation. The grid graphs, grid300-10 and grid500-10, even have such a small number of relevant pairs, that they are rounded to 0 in the precision that we choose for the numbers in the table. These low numbers of relevant pairs explain the drastic memory reduction that we achieve with our algorithm.

We can also see that for many graphs, the far-apart pairs iterator is very fast, while the hyperbolicity computation takes a long time, which is explained by the fact that we might spend quadratic time per pair to compute the hyperbolicity. Considering the percentage of far-apart pairs, we see that most graphs have roughly 30\% to 70\% far-apart pairs. The graphs with grid structure are extreme outliers, exhibiting a very low percentage of far-apart pairs. This can be explained by the grid structure, for which -- considering a grid without missing edges -- only the two pairs of opposing corners of the grid are far-apart. Furthermore, the facebook\_combined graph has a very large percentage of far-apart pairs. This is explained by the fact that it has, by far, the highest average degree of all the graphs we consider. For such a well-connected graph, BFS trees have a very large number of leaves as shortest paths are not extended further from most nodes, which is necessary to produce such a large number of far-apart pairs.

To further gain insights into the distribution of far-apart pairs, i.e., how they are distributed over various distances, we put all the distances between far-apart pairs into a histogram for four selected graphs, see \Cref{fig:hists}. While the distribution looks somewhat reasonable and expected for three of the four graphs, the notreDame-d graph shows an interesting distribution with two modes.
Note that anomalies with respect to coreness were already found in this graph~\cite{ShinEF18} where a special structure around a propeller-shaped subgraph was identified. We suspect that this structure connects a large fraction of pairs and could be related to this bi-modal observation. Finally, we also marked the part of the histogram that contains the far-apart pairs that are relevant for the hyperbolicity calculation. We observe that this region occurs after the peak of the histogram in all four instances that we show. In the notreDame-d instance, it even completely excludes the first, much larger mode.

\begin{table}[htbp]
	\caption{Time and memory consumption of the far-apart pairs iterator only. We additionally show the percentage of far-apart pairs (\enquote{far pairs}) and also the percentage of pairs which have to be considered during the hyperbolicity computation of our algorithm (\enquote{hyp pairs}). }
\label{tab:far_apart_iterator}
\centering
\begin{tabular}{|l|cccc|}
\hline
{\bfseries Graph} &	{\bfseries time (s)}	&	{\bfseries memory}	&	{\bfseries far pairs (\%)}	&	{\bfseries hyp pairs (\%)} \\
\hline
\hline

BG-MV-Physical & 21.85 & 391.23 MB & 30.91 & 4.043 \\
BG-S-Affinity\_Capture-MS & 94.32 & 1.33 GB & 36.92 & -- \\
BG-S-Affinity\_Capture-RNA & 3.29 & 124.32 MB & 70.86 & 0.602 \\
BG-S-Affinity\_Capture-Western & 23.27 & 416.42 MB & 32.42 & 5.355 \\
BG-S-Biochemical\_Activity & 1.82 & 75.47 MB & 42.95 & 8.196 \\
BG-S-Dosage\_Rescue & 0.35 & 29.13 MB & 24.67 & 5.223 \\
BG-S-Synthetic\_Growth\_Defect & 2.13 & 87.80 MB & 46.12 & 17.973 \\
BG-S-Synthetic\_Lethality & 1.08 & 54.61 MB & 42.75 & 24.658 \\
dip20170205 & 45.85 & 686.76 MB & 30.12 & 11.245 \\
\hline
CAIDA\_as\_20000102 & 4.4 & 150.73 MB & 63.94 & 5.456 \\
CAIDA\_as\_20040105 & 40.36 & 835.38 MB & 67.35 & 5.806 \\
CAIDA\_as\_20050905 & 65.66 & 1.25 GB & 69.01 & 0.502 \\
CAIDA\_as\_20110116 & 261.31 & 3.64 GB & 69.56 & 40.066 \\
CAIDA\_as\_20120101 & 361.92 & 4.39 GB & 68.86 & -- \\
CAIDA\_as\_20130101 & 438.98 & 5.02 GB & 68.89 & 5.292 \\
CAIDA\_as\_20131101 & 490.41 & 5.81 GB & 69.26 & 5.075 \\
DIMES\_201012 & 174.08 & 3.19 GB & 74.89 & 13.431 \\
DIMES\_201204 & 144.55 & 2.40 GB & 72.51 & 21.459 \\
\hline
p2p-Gnutella09 & 6.47 & 190.00 MB & 28.42 & 4.036 \\
gnutella31-d & 407.55 & 4.67 GB & 29.49 & 3.787 \\
notreDame-d & 9817 & 86.14 GB & 54.38 & 0.002 \\
\hline
ca-CondMat & 105.14 & 1.49 GB & 44.0 & 4.785 \\
ca-HepPh & 25.7 & 453.93 MB & 42.31 & 8.63 \\
ca-HepTh & 7.82 & 196.91 MB & 33.72 & 1.961 \\
com-dblp.ungraph & -- & $\skull$ & -- & -- \\
dblp-2010 & 14008 & 87.83 GB & 48.38 & 0.268 \\
email-Enron & 172.75 & 2.67 GB & 55.87 & 10.95 \\
epinions1-d & 630.49 & 7.38 GB & 45.46 & 6.683 \\
facebook\_combined & 4.5 & 158.98 MB & 89.08 & 71.779 \\
loc-brightkite & 425.11 & 4.45 GB & 36.35 & 5.019 \\
loc-gowalla\_edges & 10509 & 63.43 GB & 33.19 & 0.683 \\
slashdot0902-d & 1442.04 & 15.35 GB & 45.72 & 5.238 \\
\hline
oregon2\_010331 & 19.54 & 453.85 MB & 66.43 & 31.757 \\
t.FLA-w & -- & $\skull$ & -- & -- \\
\hline
buddha-w & $\skull$ & -- & -- & -- \\
froz-w & $\skull$ & -- & -- & -- \\
grid300-10 & 345.36 & 4.24 GB & 0.04 & 0.0 \\
grid500-10 & 2723.73 & 19.24 GB & 0.04 & 0.0 \\
z-alue7065 & 46.95 & 940.21 MB & 0.06 & 0.002 \\
\hline
\end{tabular}
\end{table}

\begin{figure}[htbp]
\begin{subfigure}{.48\textwidth}
    \centering
	\includegraphics[width=\linewidth]{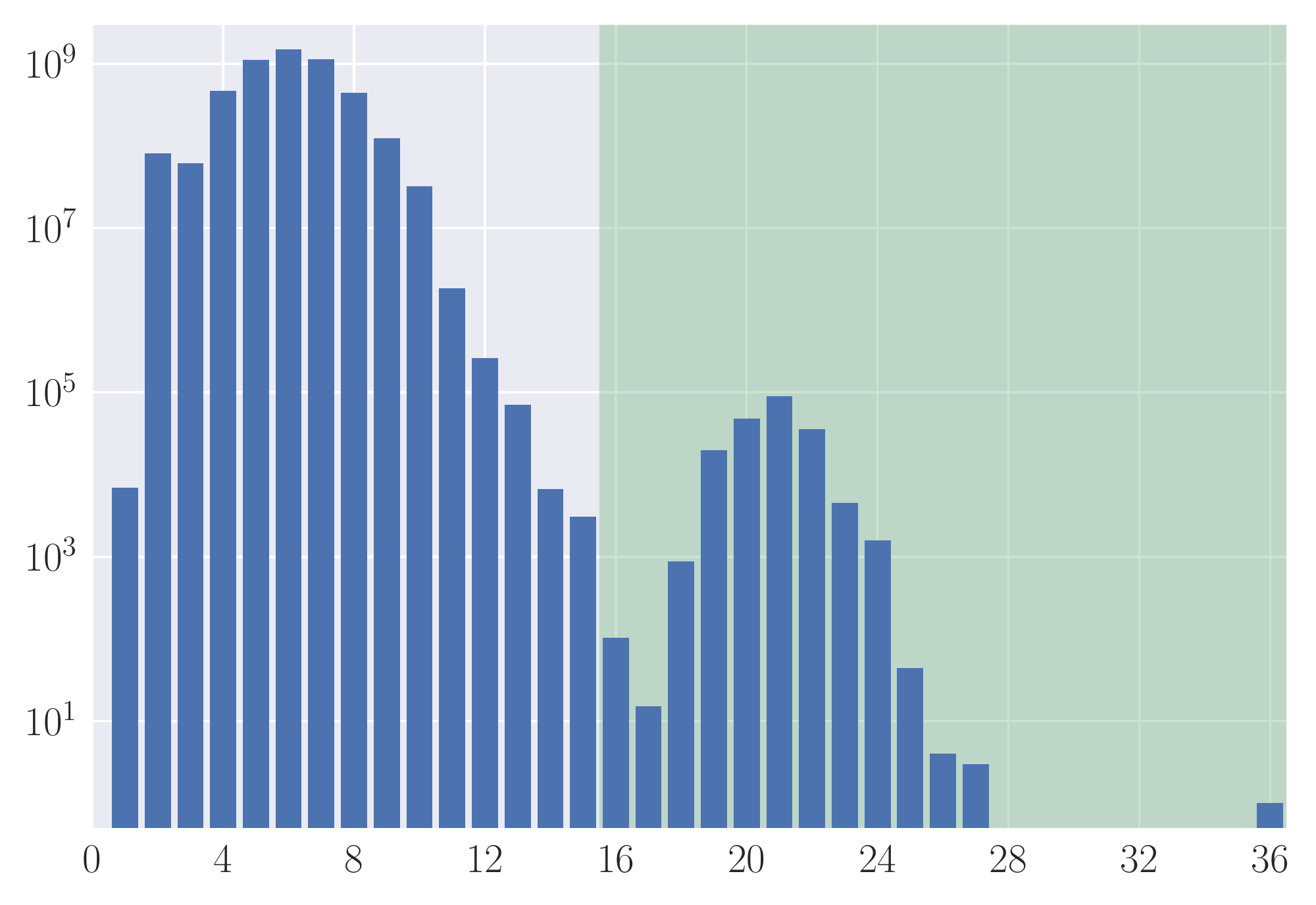}
    \caption{notreDame-d}
	\label{fig:hist_notre_dame}
\end{subfigure} \hfil
\begin{subfigure}{.48\textwidth}
    \centering
	\includegraphics[width=\linewidth]{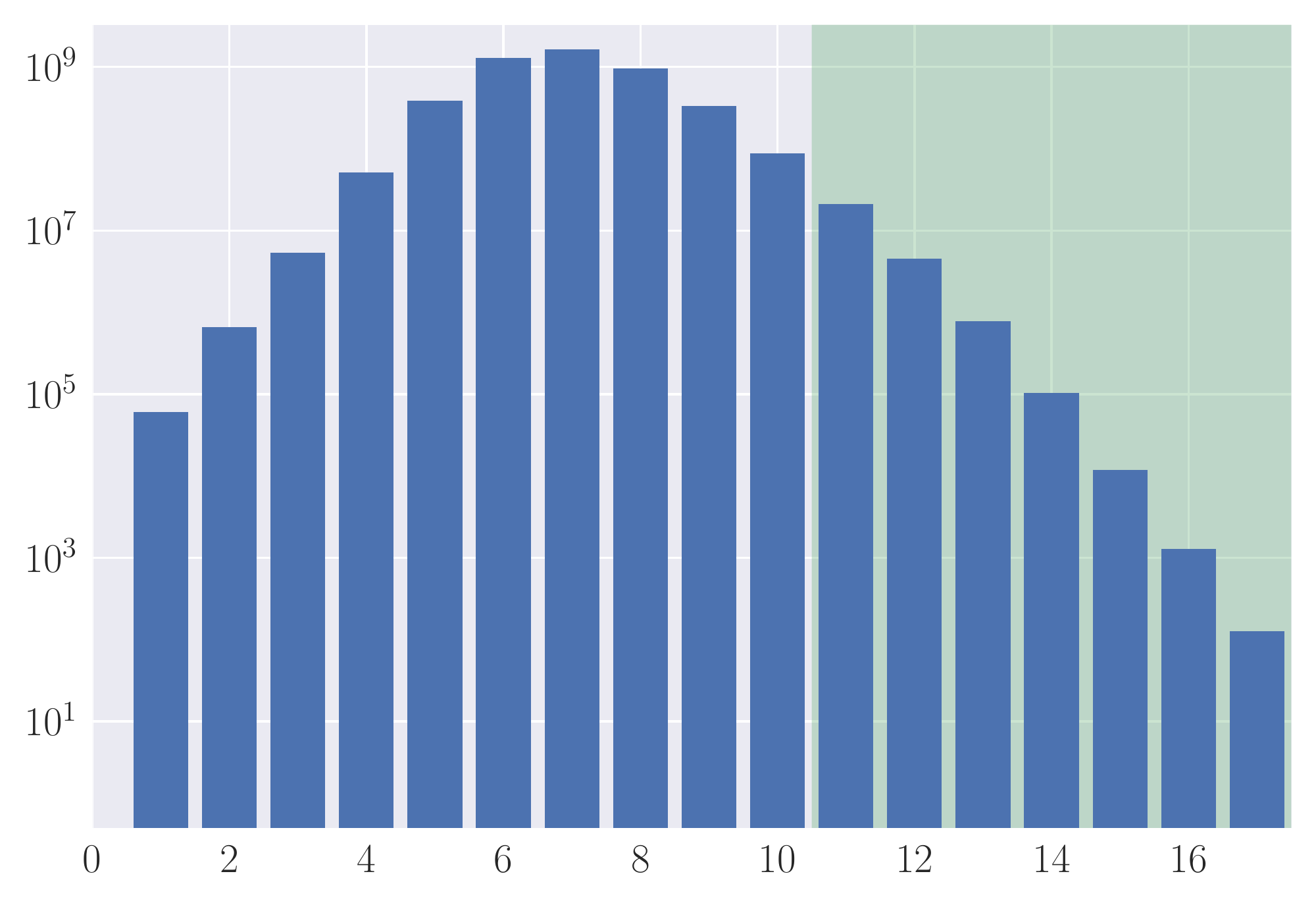}
    \caption{dblp-2010}
	\label{fig:hist_dblp}
\end{subfigure} \\
\begin{subfigure}{.48\textwidth}
    \centering
	\includegraphics[width=\linewidth]{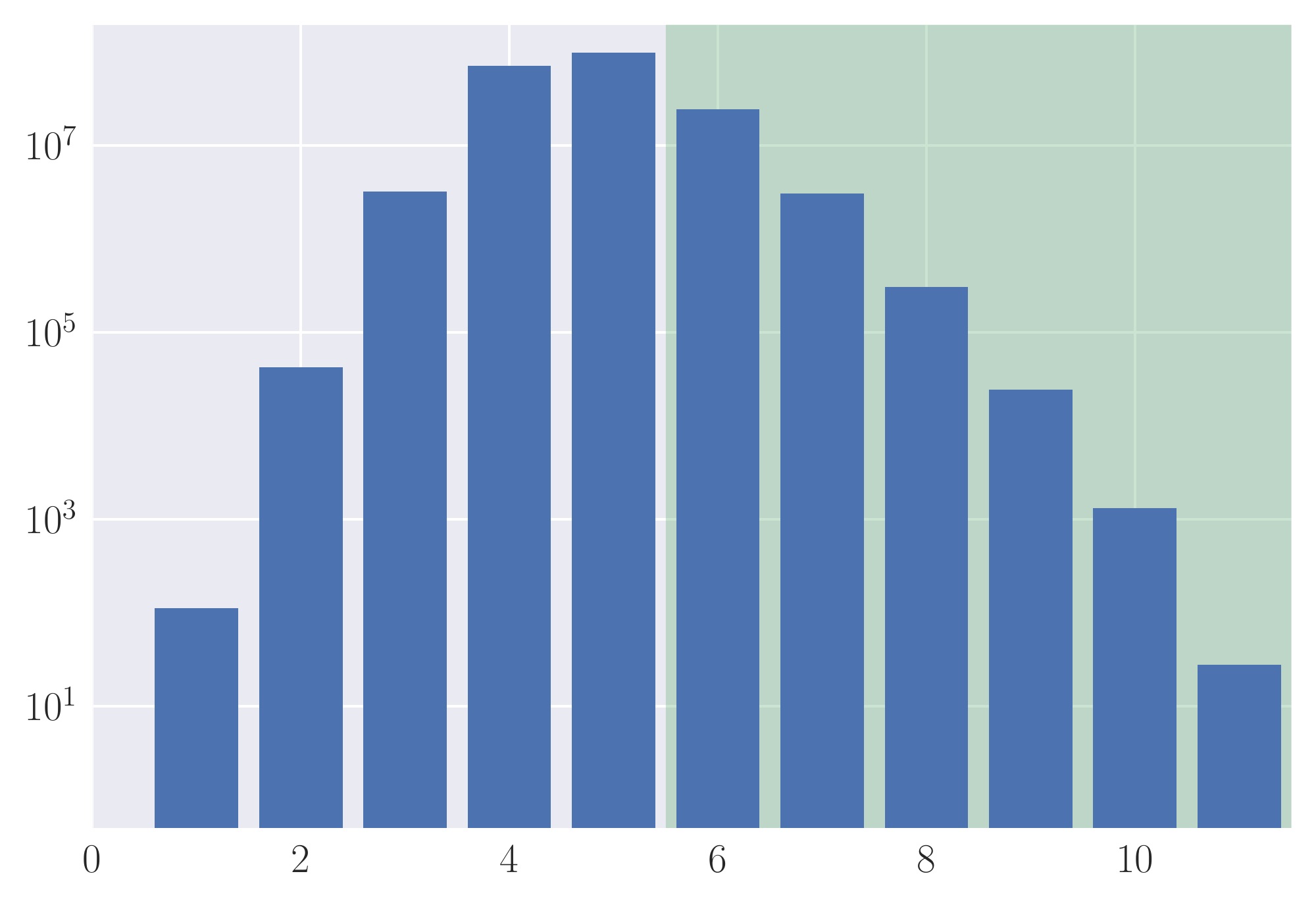}
    \caption{loc-brightkite}
	\label{fig:hist_brightkite}
\end{subfigure} \hfil
\begin{subfigure}{.48\textwidth}
    \centering
	\includegraphics[width=\linewidth]{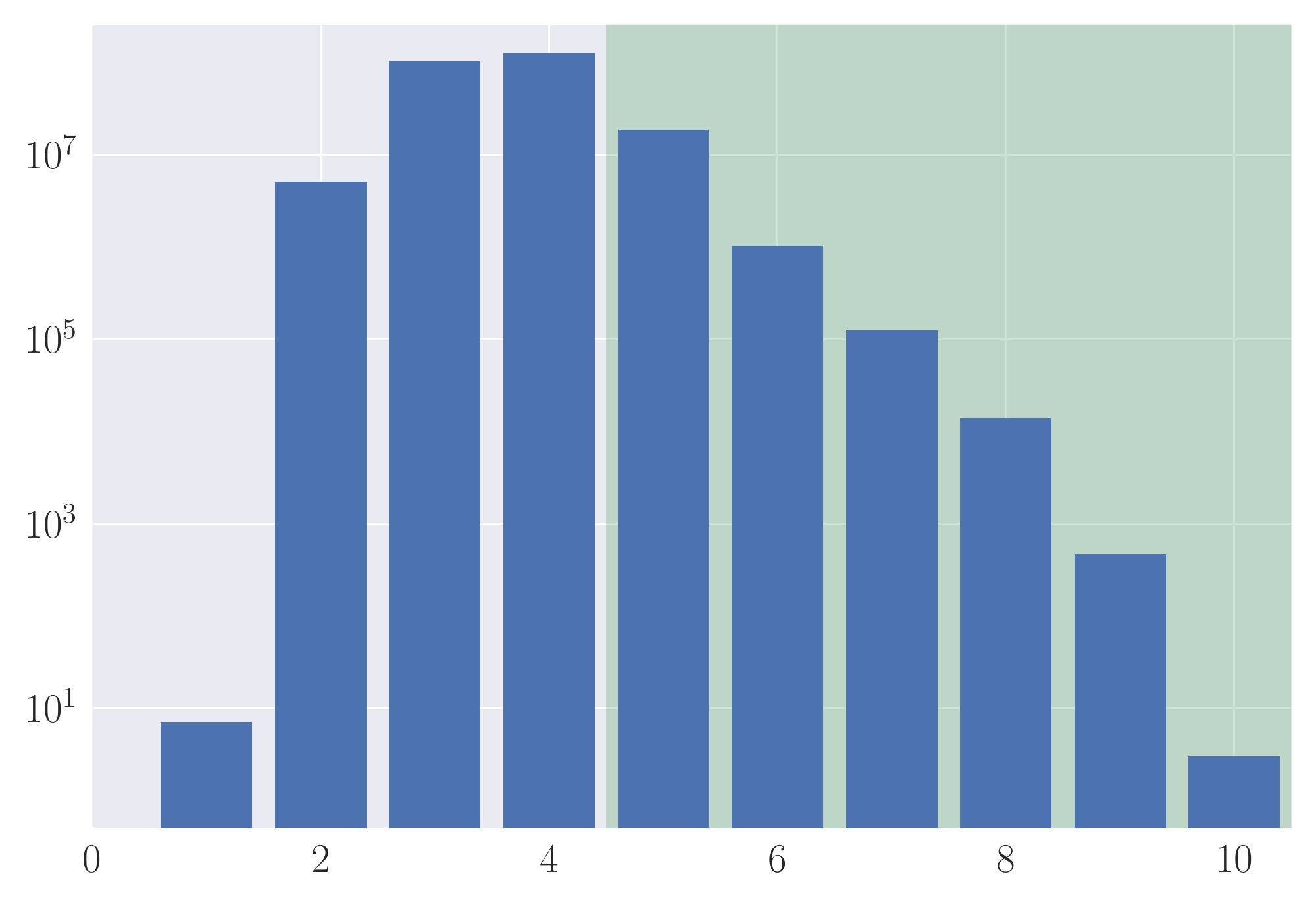}
    \caption{CAIDA\_as\_20130101}
	\label{fig:hist_caida}
\end{subfigure}
\caption{Histograms of distribution of far-apart pairs for selected graphs. On the x-axis we plot the distance and on the y-axis the number of far-apart node pairs that have this distance. The green area shows the distance range for pairs that have to be evaluated in our hyperbolicity algorithm. We can see that this area always excludes the peak of the histogram in the shown examples.}
\label{fig:hists}
\end{figure}

\section{Conclusion}
\label{sec:conclusion}
In this work, we developed a fundamental algorithm to iterate over all far-apart pairs in non increasing distance. As primary application we consider the computation of graph hyperbolicity.  Our new algorithm enables us to compute, for the first time, the hyperbolicity of some graphs with more than a hundred thousand nodes with non trivial structure (e.g., notreDame-d, loc-gowalla\_edges, com-dblp.ungraph).
We reduce the memory usage significantly, while not compromising on performance.
Non-trivial graphs with more than five hundred thousands nodes unfortunately still remain out of reach with our method. We thus plan to investigate alternative approaches in future work in order to get closer to the million nodes barrier. Furthermore, we believe that iterating over far-apart pairs in non increasing distance is such a fundamental task, that our work will enable faster algorithms also in other settings.


\bibliographystyle{plain}
\bibliography{biblio}

\appendix

\section{Time and space trade-offs for determining all far-apart pairs} \label{sec:tradeoffs}
In this section, we present algorithms offering different time and space trade-offs for the problem of computing all far-apart pairs. The space complexity considered here is the working memory, hence excluding the space needed to store the result, unless needed during computations.

A first algorithm to determine the set of far-apart pairs is to: 1) determine for each vertex $u\in V$ the set $F_u$ of $u$-far vertices using BFS, and then 2) check for each vertex $v\in F_u$ if $u$ is $v$-far (i.e., if $u \in F_v$). This can be done in time $\OO(nm)$ and space $\OO(n^2)$ since $|F_u| \in \OO(n)$.

Another algorithm is to execute two steps for each vertex $u\in V$:
1) determine the set $F_u$ of $u$-far vertices using BFS, and then 2) for each vertex $v\in F_u$, check if there is $w\in N(u)$ such that $\dd(u,v) < \dd(w,v)$. The second step requires to compute distances from $w$ and so the time complexity of this algorithm is $\OO((n + m)\sum_{u\in V}(1 + |N(u)|))=\OO(m^2)$. Observe that during the processing of each vertex $u$, this algorithm stores the set $F_u$, the distances from $u$, and the distances from one neighbor $w$ of $u$. Hence, the space complexity is $\OO(n)$.
The second method can be improved using the bit-parallel BFS proposed in~\cite{AkibaIY13}. Indeed, this algorithm computes simultaneously distances from $u$ and $b$ of its neighbors in time $\OO(n+m)$ and space $\OO(n)$, assuming that $b$ is a constant and using bitwise operations on bit vectors of size $b$ (typically 32 or 64). Hence, the number of BFSs to perform is reduced to $\sum_{u\in V}\left\lceil |N(u)| /b \right\rceil$.

Let us now show how to modify the above algorithm to obtain an algorithm with time complexity in $\OO(nm)$ and space complexity in $\OO(n\pw(G))$, where $\pw(G)$ denotes the pathwidth of $G$~\cite{dias2002,DBLP:journals/eatcs/Petit11}. The main idea is to compute the distances from $u$ only once and to store them during as few iterations as possible. For that, let $\pi: V\to [n]$ be a linear ordering of the vertices (i.e., a bijective mapping) related to the pathwidth as described later and let $\Pi(V)$ be the set of all such orderings. The algorithm iterates over the vertices in the order $\pi^{-1}(1), \pi^{-1}(2), \ldots, \pi^{-1}(n)$. The distances from $u$ are used for the processing of vertex $u$, and for the processing of each neighbor $w\in N(u)$.
Let $w_L \coloneqq \arg\min_{w \in N(u)} \pi(w)$ and $w_R \coloneqq \arg\max_{w \in N(u)} \pi(w)$.
Hence, distances from $u$ must be computed at iteration $\pi(w_L)$ and stored until iteration $\pi(w_R)$.
Consequently, at iteration $i$, we have stored distances from all the vertices in
\[
\{u\in V\mid \exists uv\in E\textrm{ s.t. } \pi(u)< i \leq \pi(v)\} \cup \{\pi^{-1}(i)\} \cup \{v\in V\mid \exists uv\in E\textrm{ s.t. } \pi(u)\leq i \leq \pi(v)\}.
\]
Now observe that the \emph{pathwidth}~\cite{dias2002,DBLP:journals/eatcs/Petit11} of a graph $G$ is defined as $\pw(G) = \min_{\pi\in \Pi(V)} p(G, \pi)$, where $p(G, \pi) = \max_{i=1}^n |\{u\in V\mid \exists uv\in E\textrm{ such that } \pi(u)< i \leq \pi(v)\}|$. 
Hence, at iteration $i$ we store distances from at most $2p(G,\pi)+1$ vertices and there is an ordering ensuring to store distances from at most $2\pw(G)+1$ vertices.
Consequently, the space complexity of this algorithm is in $\OO(n\pw(G))$ and its time complexity is $\OO(nm)$ assuming that the ordering $\pi$ such that $p(G,\pi) = \pw(G)$ is given. However, the problem of determining an ordering $\pi$ such that $p(G,\pi)=\pw(G)$ is NP-complete~\cite{dias2002,DBLP:journals/eatcs/Petit11}. Nonetheless, efficient heuristic algorithms have been proposed~\cite{DBLP:journals/jea/CoudertMN16,Castillo2018}.

Finally, recall that the \emph{bandwidth} of a graph $G$ is defined as $\bw(G) = \min_{\pi\in \Pi(V)} b(G,\pi)$, where $b(G, \pi)=\max_{uv\in E}|\pi(u)-\pi(v)|$~\cite{dias2002,DBLP:journals/eatcs/Petit11}. Therefore, distances from $u$ are stored for at most $2b(G, \pi)+1$ iterations, and there is an ordering ensuring that this number of iterations is at most $2\bw(G)+1$. However, the problem of determining such an ordering is NP-hard~\cite{GareyJohnson1979}.

\section{Retrieving distances from far vertices}
\label{sec:distfromfarvertices}

First note the following corollary which is a direct consequence of \Cref{lem:far}:

\begin{corollary} \label{cor:meanleaf}
For any $v \in V$, let $F_v$ be the set of $v$-far vertices. The number of leaves of any shortest path tree rooted at $v$ is at least $|F_v|$.
\end{corollary}

Observe however that the lower bound of \Cref{cor:meanleaf} does not imply the existence of a shortest path tree with $|F_v|$ leaves. For instance, consider the 4-cycle $(u_1, u_2, u_3, u_4)$. We have $|F_{u_1}| = |\set{u_3}| = 1$, but all shortest path trees rooted at $u_1$ have 2 leaves (either $\set{u_2, u_3}$ or $\set{u_3, u_4}$).

We now show that, given the $v$-far vertices and a constant $c$, one can determine the vertices at distance at least $c$ from $v$.

\begin{lemma}
Let $v \in V$ and let $S$ be the set of all vertices $u$ for which $\dd(u,v) \geq c$. Given the set $F_v$ of $v$-far vertices and the distances from $v$ to each of these vertices, one can compute the set $S$ in time $\mathcal{O}(\lvert F_v \rvert +  \sum_{u \in S} \lvert N(u) \rvert)$.
\end{lemma}

\begin{proof}
First, observe that if $u$ is $v$-far, then for each $w\in N(u)$ it holds that
\[
\dd(v, u) -1 \leq \dd(v, w)\leq \dd(v, u).
\]
Furthermore, a neighbor $w \in N(u)$ with $\dd(v, w) = \dd(v, u) - 1$ is not $v$-far, and if a neighbor $w' \in N(u)$ is $v$-far then $\dd(v, u) = \dd(v, w')$.

To determine all the nodes that are at distance at least $c$ from $v$, it suffices to perform a reverse breadth-first search, starting from far vertices.
More precisely, let $\{L_d\}_{d \in \{c, \dots, \ecc(v)+1\}}$ be a set family where $L_d$ is initialized with the $v$-far vertices at distance $d$ from $v$, and let $L_{\ecc(v) + 1} = \emptyset$.
Then, consider these sets in decreasing distance $d$ from $v$, starting with $d = \ecc(v)$, and stopping when $d = c-1$.
For each vertex $u \in L_d$, add each $w\in N(u)$ to $L_{d-1}$ for which $w\not\in L_{d}\cup L_{d+1}$.
At the end of this procedure, we have that if $\dd(v,u) = d$, then $u\in L_d$ and all the sets are disjoint, i.e. $L_d \cap L_{d'} = \emptyset$ for any $c\leq d,d'\leq \ecc(v)$ with $d \neq d'$. The claimed time bound follows immediately from the description of the algorithm.
\end{proof}

\end{document}